\documentclass[11pt,reqno]{amsart}
\usepackage{amsmath,amssymb,graphicx,mathrsfs,amsopn,stmaryrd,amsthm,color,bm,times}
\usepackage{dynkin-diagrams}[root radius = 2cm]
\usepackage{ytableau}
\usepackage[colorlinks=true,citecolor=blue,linkcolor=blue,pagebackref]{hyperref}
\usepackage[final]{showlabels}
\usepackage[english]{babel}
\usepackage[mathscr]{euscript}
\usepackage{xcolor}
\usepackage{geometry}
\let\emptyset\varnothing
\usepackage{tikz}    
\usetikzlibrary{arrows,matrix,decorations.markings,shapes.geometric}   
\usetikzlibrary{decorations.pathreplacing}

\numberwithin{equation}{section}
\newtheorem{thm}{Theorem}[section]
\newtheorem{prop}[thm]{Proposition}
\newtheorem{lem}[thm]{Lemma}
\newtheorem{cor}[thm]{Corollary}
\newtheorem{conj}[thm]{Conjecture}
\theoremstyle{definition}

\theoremstyle{remark}
\newtheorem{rem}[thm]{Remark}

\newcommand{\beq}{\begin{equation}}
\newcommand{\eeq}{\end{equation}}
\newcommand{\be}{\begin{equation*}}
\newcommand{\ee}{\end{equation*}}
\newcommand{\bs}{\boldsymbol}

\newcommand{\C}{\mathbb{C}}

\newcommand{\Z}{\mathbb{Z}}

\newcommand{\mc}{\mathcal}
\newcommand{\cD}{\mathcal{D}}
\newcommand{\cF}{\mathcal{F}}
\newcommand{\cK}{\mathcal{K}}
\newcommand{\cR}{\mathcal{R}}

\newcommand{\g}{\mathfrak{g}}
\newcommand{\gl}{\mathfrak{gl}}
\newcommand{\h}{\mathfrak{h}}
\newcommand{\n}{\mathfrak{n}}
\newcommand{\osp}{\mathfrak{osp}}
\newcommand{\fksp}{\mathfrak{sp}}
\newcommand{\so}{\mathfrak{so}}
\newcommand{\fkb}{\mathfrak{b}}

\newcommand{\fkS}{\mathfrak{S}}

\newcommand{\Wr}{\mathrm{Wr}}

\newcommand{\ord}{\mathrm{ord}\,}
\newcommand{\id}{{\mathrm{id}}}   
   
\newcommand{\bsing}{{\tl{\bm s}\mathrm{ing}}}

\newcommand{\sF}{\mathscr{F}}

\newcommand{\pa}{\partial}
\newcommand{\tl}{\tilde}
\newcommand{\tlbms}{{\tilde{\bm s}}}
\newcommand{\gge}{\geqslant}
\newcommand{\lle}{\leqslant}
\newcommand{\la}{\lambda}
\newcommand{\La}{\Lambda}

\newcommand{\bla}{\bm\lambda}

\newcommand{\glMN}{\mathfrak{gl}_{m|n}}

\newcommand{\bmx}{\begin{pmatrix}}    
\newcommand{\emx}{\end{pmatrix}}   
\newcommand{\wt}{\widetilde}    

\newcommand{\ve}{\varepsilon}

\begin{document}
\pagestyle{myheadings}
\setcounter{page}{1}

\title[Bethe ansatz equations for orthosymplectic Lie superalgebras]{Bethe ansatz equations for orthosymplectic Lie superalgebras\\ and self-dual superspaces}

\author{Kang Lu and Evgeny Mukhin}
\address{K.L.: Department of Mathematics, University of Denver, 
\newline
\strut\kern\parindent 2390 S. York St., Denver, CO 80208, USA}\email{kang.lu@du.edu}
\address{E.M.: Department of Mathematical Sciences,
Indiana University-Purdue University\newline
\strut\kern\parindent Indianapolis, 402 N.Blackford St., LD 270,
Indianapolis, IN 46202, USA}\email{emukhin@iupui.edu}

\begin{abstract} 
	We study solutions of the Bethe ansatz equations associated to the orthosymplectic Lie superalgebras $\osp_{2m+1|2n}$ and $\osp_{2m|2n}$. Given a solution, we define a reproduction procedure and use it to construct a family of new solutions which we call a population. To each population we associate a symmetric rational pseudo-differential operator $\mc R$. Under some technical assumptions, we show that the superkernel $W$ of $\mc R$ is a self-dual superspace of rational functions, and the population is in a canonical bijection with the variety of isotropic full superflags in $W$ and with the set of symmetric complete factorizations of $\mc R$. In particular, our results apply to the case of even Lie algebras of type D${}_m$ corresponding to $\osp_{2m|0}=\mathfrak{so}_{2m}$.
	
		\medskip 
		
\end{abstract}

\maketitle


\thispagestyle{empty}
\section{Introduction}	
The Gaudin models associated to simple Lie algebras have been a subject of intensive research for half a century which produced a number of spectacular results. Then the importance of the supersymmetric generalization has been understood and Gaudin models related to simple Lie superalgebras came to attention of many mathematicians and physicists. In this paper, we study the Bethe ansatz equations (BAE) of the Gaudin models associated with orthosymplectic Lie superalgebras  $\osp_{2m+1|2n}$ and $\osp_{2m|2n}$.

To our knowledge, most of the previous work in this direction has been done for the case of $\osp_{1|2}$, see \cite{KM01,Zei15}, with a few papers devoted to the general case for the XXX spin chains, \cite{AACDFR,Tsu99}. 

Our guidance and motivation come from \cite{MV04} where the BAE are studied for types B and C and \cite{HMVY19} which treats the BAE for the Lie superalgebra $\glMN$. 

The BAE is a system of algebraic equations which depends on the Cartan matrix, a choice of highest weights, and evaluation parameters. It is widely expected that the solutions of the BAE produce eigenvectors of the Gaudin Hamiltonians acting in the tensor product of the corresponding highest weight modules. The Gaudin Hamiltonians are produced from the action of a remarkable commutative subalgebra of the universal enveloping algebra of the current Lie algebra, called the Bethe subalgebra, see for example \cite{F, KM01, MM, MTV09, MVY15}. We do not discuss the Gaudin Hamiltonians or their eigenvectors in this paper and concentrate on the study of the solution set of the BAE. Our main tool is a reproduction procedure, which given a solution of the BAE and a choice of a simple root, produces a family of new solutions of the BAE. 

An important difference of the super case is that we have a number of different Cartan matrices corresponding to the choices of the Borel subalgebras. In the orthosymplectic case, the choices of the Cartan matrix $C=(c_{ij})$ are parameterized by parity sequences $\bs s=(s_1,\dots,s_r)$, $s_i\in\{\pm 1\}$, where the number of positive ones is $m$ and $r$ is the rank of the Lie superalgebra. In the case of $\osp_{2m|2n}$ and $s_r=-1$ for each Cartan matrix we have two Borel subalgebras which are recorded by $\kappa\in\{\pm1\}$ while in the case of $s_r=1$ or of $\osp_{2m+1|2n}$ such a Borel subalgebra is unique, recorded by $\kappa=1$. We call $\tlbms=(\bm s;\kappa)$ an extended parity sequence.

A solution of the BAE is given by zeroes of a sequence of polynomials in one variable $\bs y=(y_1,\dots,y_r)$. Here the degree of $y_i$ is the number of unknowns in the BAE of color $i$. The reproduction procedure produces a family of new sequences $\bs y^{[i]}:=(y_1,\dots,\tilde y_i,\dots,y_r)$.

There are two different cases of the reproduction procedure: the case of roots of non-zero length  and the case of roots of length zero. In the former case, the new polynomial $\tilde y_i$ is described by the first order differential equation,
$$
y_i'\tilde y_i-y_i\tilde y_i'=T_i\prod_{j,\, j\neq i} y_{j}^{-c_{ij}},
$$
and we call such reproduction procedure bosonic.
In the latter case $\tilde y_i$ is obtained simply by division
$$
y_i\tilde y_i= \ln'\Big(T_i\prod_{j}y_{j}^{-c_{ij}}\Big) \pi_i \prod_{j,\, c_{ij}\ne 0} y_j,
$$
and in this case we call such reproduction procedure fermionic. Here the polynomials $T_i$ ($T_i$ do depend on the extended parity sequence $\tlbms$) correspond to the highest weights and evaluation parameters of the modules, and $\pi_i$ is the denominator of the rational function $T_i'/T_i$ with minimal degree.

Quite generally, the bosonic reproduction procedure produces a one-parameter family of solutions of the BAE with the same Cartan matrix and polynomials $T_i$ (but with a different number of unknowns). At the same time the fermionic reproduction procedure produces a single sequence $\bs y^{[i]}$ which is (provided $y_{i-1},y_{i+1}$ have no multiple roots, no common roots, and no roots at evaluation parameters)
a solution of the BAE with the new Cartan matrix according to $\bs s^{[i]}:=(s_1,\dots,s_{i+1},s_i,\dots,s_r)$ for $i<r$  (and $\bm s^{[r]}=\bm s^{[r-1]}$) and new polynomials $T_i$ determined by some rule, see Table \ref{table:1}. The set of all sequences obtained by successive applications of the reproduction procedure (paired up with the corresponding extended parity sequences $\tlbms$) is called a population which is the main object of our study.

\medskip

The main idea is to relate a population associated with an orthosymplectic Lie superalgebra to a population associated with a general linear Lie superalgebra. We observe that if $\bs y$ is a solution of the $\osp_{2m+1|2n}$ BAE then the sequence of polynomials $(y_1,\dots,y_{r-1},y_r,y_{r-1},\dots, y_1)$ is a solution of the $\gl_{2m|2n}$ BAE. We also note that if $\bs y$ is a solution of the $\osp_{2m|2n}$ BAE with a Cartan matrix of type D, then the sequence of polynomials $(y_1,\dots,y_{r-2},y_{r-1}y_r,y^2_r,y_r^2,y_{r-1}y_r,y_{r-2},\dots, y_1)$ is a (generalized) solution of the $\gl_{2m|2n+1}$ BAE. We remark that $\osp_{2m+1|2n}$ and $\osp_{2m|2n}$ are not subalgebras of $\gl_{2m|2n}$ and $\gl_{2m|2n+1}$, respectively. However, the twining characters of the latter Lie algebras give the characters of the former, see \cite{FSS96} and \cite[Sections 8 \& 9]{KK00}, cf. \cite[Appendix]{LMV17}.

Motivated by these observations, we introduce a rational pseudo-differential operator $\mc R$ which does not change under the reproduction procedure. For the case of $\osp_{2m+1|2n}$, it has the form
$$
\mc R=\mathop{\overrightarrow\prod}\limits_{1\lle i\lle r}\Big(\pa-s_i\ln'\frac{P_iy_{i-1}}{y_i}\Big)^{s_i}\mathop{\overleftarrow\prod}\limits_{1\lle i\lle r}\Big(\pa+s_i\ln'\frac{P_iy_{i-1}}{y_i}\Big)^{s_i},
$$
where $y_0=1$ and $P_i$ are rational functions related to $T_i$, see \eqref{eq:T-polynomials}. The formula for $\mc R$ in the case of  $\osp_{2m|2n}$ is given in \eqref{eq:D-oper}. 

We write $\mc R$ as a ratio of differential operators $\mc R=\mc D_{\bar 0} \mc D_{\bar 1}^{-1}$. We expect that the coefficients of $\mc D_{\bar 0}$ and $\mc D_{\bar 1}$ are eigenvalues of a special set of generators of the Bethe subalgebra acting on the eigenvector corresponding to the solution of the BAE, see  \cite{LM20, MTV06, MM, LM19}.

We consider the superkernel $W=\ker(\mc D_{\bar 0})\oplus\ker(\mc D_{\bar 1})$. It turns out that $W$ has some remarkable properties and, with some technical assumptions, $W$ consists of rational functions. Let us describe some of these properties. First, the Wronski determinant is one, $\Wr(W)=1$. Second, there exists an isomorphism of vector spaces $\phi: \Lambda^{\dim W-1}(W)\to W$ between the exterior power of $W$ and $W$ given by $\phi(H)=\Wr(H)$ for all $H\in\Lambda^{\dim W-1}(W)$. Third, we have the non-degenerate bilinear form $(\ ,\ ): W\otimes W\to \C$ given by $(w_1,w_2)=\Wr(\phi^{-1}(w_1),w_2)$. This bilinear form is symmetric if $\dim W$ is odd and skew symmetric if $\dim W$ is even. Fourth, we have $\ker(\mc D_{\bar 0}) \perp \ker(\mc D_{\bar 1})$ with respect to this bilinear form. Fifth, for any subspace $\mathscr V\subset W$, we have $\Wr(\mathscr V)=\Wr( \mathscr V^\perp)$. It means that the space $W$ is essentially, the self-dual space, studied in \cite{Lu18,LMV17,MV04}. Taking into account the fourth property, we call $W$ the self-dual superspace.

Then we show that the set of all full isotropic superflags in $W$ is in a natural bijection with the population and with the set of all symmetric complete factorizations of $\mc R$. We note that for the Cartan matrices of type D, present in the case of $\osp_{2m|2n}$, a point of the population corresponds to two isotropic superflags related by what we call the ``fake" reproduction procedure.

In particular, the special cases $\osp_{2m+1|0}=\mathfrak{so}_{2m+1}$, $\osp_{0|2n}=\mathfrak{sp}_{2n}$,  recover the known even cases of types B$_m$ and C$_n$ of \cite{MV04} and $\osp_{2m|0}=\mathfrak{so}_{2m}$ treats the even case of type D$_m$, see Section \ref{sec:even}.

It is important to mention that the bijection between the population and the variety of isotropic superflags is proved under the technical assumption of superfertility. Namely, one assumes that the reproduction procedure can be repeated indefinitely even if some of the produced sequences of polynomials are not solutions of the Bethe ansatz equations. We expect that superfertility is a generic feature but an additional study of obstructions to the reproduction procedure is needed. We discuss this issue in Section \ref{sec:superfertility}.

The results of this paper are a step towards connecting the Bethe ansatz to algebraic geometry. Such a connection which is finalized in \cite{MTV09} for Gaudin models associated to $\gl_N$, is a long term motivation. One expected outcome is the understanding of perfectly integrable nature of Gaudin models, see \cite{Lu20}.

In the paper we work with periodic boundary conditions. One can generalize our results to quasi-periodic conditions in a straightforward way, cf. \cite{MV08}. We expect that the main statements of this paper can be also obtained for the XXX type BAE related to orthosymplectic Lie superalgebras, using \cite{MV03,HLM}.

\bigskip

The paper is organized as follows. In Section \ref{osp sec} we discuss the details of the root systems and general facts for orthosymplectic Lie superalgebras. Section \ref{operators sec} is devoted to the study of the symmetric rational pseudo-differential operators and self-dual superspaces. The main statements are Lemma \ref{lem:super-self-dual} and Proposition \ref{prop:bij-f-iso}.  We introduce the BAE and related terminology in Section \ref{sec: bethe}. In Section \ref{sec: reproduction} we study the reproduction procedure. We describe a general fermionic reproduction in Table \ref{table:1} and then apply it to the cases of $\osp_{2m+1|2n}$ in Section \ref{sec b} and of $\osp_{2m|2n}$ in Section \ref{sec D}. In particular, we introduce the symmetric rational pseudo-differential operators associated to a population in Sections \ref{sec B oper} and \ref{sec D2}. Section \ref{sec: main} contains our main result, Theorem \ref{main thm}.

\bigskip

{\bf Acknowledgments.}
This work was partially supported by the Simons Foundation grants \#353831 and \#709444.

\section{Preliminaries on orthosymplectic Lie superalgebras}\label{osp sec}
Assume $m,n\in\Z_{\gge 0}$. Let $\iota$ be either $0$ or $1$. In this section, we recall the basics for general Lie superalgebras $\gl_{m|n}$ and orthosymplectic Lie superalgebras $\osp_{2m+\iota|2n}$. For details, see e.g. \cite{CW12}. 

Set $r=m+n$ and $I=\{1,\dots,r\}$. For $k\in\Z$, we set $k^\circ=k+2m+\iota$. 

\subsection{General linear Lie superalgebras}
A \emph{vector superspace} $W = W_{\bar 0}\oplus W_{\bar 1}$ is a $\Z_2$-graded vector space. We call elements of $W_{\bar 0}$ \emph{even} and elements of
$W_{\bar 1}$ \emph{odd}. We write $|w|\in\{\bar 0,\bar 1\}$ for the parity of a homogeneous element $w\in W$. Set $(-1)^{\bar 0}=1$ and $(-1)^{\bar 1}=-1$.

Let $\C^{m|n}$ be a complex vector superspace, with $\dim(\C^{m|n})_{\bar 0}=m$ and $\dim(\C^{m|n})_{\bar 1}=n$.  Choose a homogeneous basis $e_i$, $i\in I$, of $\C^{m|n}$ such that $|e_i|=\bar 0$ for $1\lle i\lle m$ and $|e_i|=\bar 1$ for $m+1\lle i\lle r$. We call it the \emph{standard basis} of $\C^{m|n}$. Set $|i|=|e_i|$.

Let $\bm s=(s_1,\dots,s_{r})$ where $s_i\in \{\pm 1\}$ and $s_i=1$ exactly $ m$ times. We call such a sequence $\bm s$ a \emph{parity sequence}. Denote the set of all parity sequences by $S_{ m| n}$.  Let $\mathfrak S_k$ be the symmetric group permuting elements in $\{1,\dots,k\}$. For each $\bm s\in S_{ m| n}$, define $\sigma_{\bm s}\in \mathfrak S_{ r}$ by
\[
\sigma_{\bm s}(i)=\begin{cases}
\#\{j~|~j\lle i,~ s_j=1\}, \qquad &\text{ if }s_i=1,\\
 m+\#\{j~|~j\lle i,~ s_j=-1\}, \qquad &\text{ if }s_i=-1.
\end{cases}
\]
Set $\bm s_+=(1,\dots,1,-1,\dots,-1)$ and $\bm s_-=(-1,\dots,-1,1,\dots,1)$.
We have $\sigma_{\bm s_+}=\id$. Also $\sigma_{\bm s_-}(i)=i+ m$ when $i\lle  n$, and $\sigma_{\bm s_-}(i)=i-n$ when $i>n$.

For $\bm s\in S_{ m|n}$ and $i\in I$, define numbers
\[
\bm s_i^+=\#\{j~|~j>i,~s_j=1\},\qquad \bm s_i^-=\#\{j~|~j<i,~s_j=-1\}.
\]
The relations between $\sigma_{\bm s}$ and $\bm s_i^+,\bm s_i^-$ are given by
\[
\bm s_i^+=\begin{cases}
 m-\sigma_{\bm s}(i),\quad &\text{ if }s_i=1,\\
\sigma_{\bm s}(i)-i,\quad &\text{ if }s_i=-1,\\
\end{cases}\quad
\bm s_i^-=\begin{cases}
i-\sigma_{\bm s}(i),\quad &\text{ if }s_i=1,\\
\sigma_{\bm s}(i)- m-1,\quad &\text{ if }s_i=-1.\\
\end{cases}
\]

The Lie superalgebra $\glMN$ is generated by elements $e_{ij}$, $i,j\in I$, with the supercommutator relations
\[
[e_{ij},e_{kl}]=\delta_{jk}e_{il}-(-1)^{(|i|+|j|)(|k|+|l|)}\delta_{il}e_{kj},
\]
and the parity of $e_{ij}$ is given by $|i|+|j|$. 

The \emph{Cartan subalgebra $\h$} of $\glMN$ is spanned by $e_{ii}$, $i\in I$. Let $\ve_i$, $i\in I$, be a basis of $\h^*$ (the dual space of $\h$) such that $\ve_i(e_{jj})=\delta_{ij}$. There is a bilinear form $(\ ,\ )$ on $\h^*$ given by $(\ve_i,\ve_j)=(-1)^{|i|}\delta_{ij}$. The \emph{root system $\bf{\Phi}$} is a subset of $\h^*$ given by
\[
{\bf \Phi}:=\{\ve_i-\ve_j~|~i,j\in I \text{ and }i\ne j\}.
\]
We call a root $\ve_i-\ve_j$ \emph{even} (resp. \emph{odd}) if $|i|=|j|$ (resp. $|i|\ne |j|$).

A $\glMN$-weight $\la$ is called \textit{dominant integral} if $(\la,\ve_i)\in\Z$ for all $1\lle i\lle r$ and $(-1)^{|j|}(\la,\ve_j-\ve_{j+1})\gge 0$ for all $1\lle j\lle r-1$ except possibly $j=m$. A $\glMN$-weight $\la$ is called \textit{typical} if 
\[
(\la,\ve_i-\ve_j)-i-j+1+2m\ne 0,
\]
for all $1\lle i\lle m < j\lle r$.

Given a parity sequence $\bm s\in S_{m|n}$, we define the set of $\bm s$-\emph{positive roots} $\Phi_s^+=\{\ve_{\sigma_{\bm s}(i)}-\ve_{\sigma_{\bm s}(j)}~|~i,j\in I \text{ and }i<j\}$. Define the $\bm s$-\emph{simple positive roots} by $\alpha_i^{\bm s}=\ve_{\sigma_{\bm s}(i)}-\ve_{\sigma_{\bm s}(i+1)}$, $1\lle i\lle r-1$.

Note that $(-1)^{|\sigma_{\bm s}(i)|}=s_i$.

The \emph{symmetrized Cartan matrix} $B^{\bm s}:=(b_{ij}^{\bm s})_{1\lle i,j\lle r-1}=\big((\alpha_i^{\bm s},\alpha_j^{\bm s})\big)_{1\lle i,j\lle r-1}$ for $\glMN$ associated with the parity sequence $\bm s$ is an $(r-1)\times (r-1)$ matrix described by the $2\times 2$ submatrices
\beq\label{eq:cartan-block}
\begin{pmatrix}
(\alpha_i^{\bm s},\alpha_{i}^{\bm s}) & (\alpha_i^{\bm s},\alpha_{i+1}^{\bm s})\\
(\alpha_{i+1}^{\bm s},\alpha_i^{\bm s}) & (\alpha_{i+1}^{\bm s},\alpha_{i+1}^{\bm s})	
\end{pmatrix}
=
\begin{pmatrix}
s_{i}+s_{i+1} & -s_{i+1}\\
-s_{i+1} & s_{i+1}+s_{i+2}	
\end{pmatrix},
\eeq
where $1\lle i\lle r-2$. 

The \emph{Cartan matrix} $C^{\bm s}:=(c_{ij}^{\bm s})_{1\lle i,j\lle r-1}$ for $\glMN$ associated with the parity sequence $\bm s$ is given by $c_{ij}^{\bm s}=s_{i}b_{ij}^{\bm s}$. In particular, we have
$c_{ii}=2$ if $s_i=s_{i+1}$, and $c_{i+1,i}=-1$ for all $i$.

\subsection{Orthosymplectic Lie superalgebras}
Let $V=V_{\bar 0}\oplus V_{\bar 1}$ be a complex vector superspace. We call a bilinear form $(\cdot,\cdot): V \otimes V\to V$ \emph{even} if $(V_i,V_j)=0$ unless $i+j=\bar 0$. We call an even bilinear form $(\cdot,\cdot)$ \emph{supersymmetric} if $(\cdot,\cdot)|_{V_{\bar 0}\otimes V_{\bar 0}}$ is symmetric and $(\cdot,\cdot)|_{V_{\bar 1}\otimes V_{\bar 1}}$ is skew-symmetric.

Let $(\cdot,\cdot)$ be an even nondegenerate supersymmetric bilinear form on the vector superspace $\C^{2m+\iota|2n}$. Set
\[
\osp_{2m+\iota|2n}:=\langle g\in \gl_{2m+\iota|2n}~|~(gx,y)+(-1)^{|g|\cdot |x|}(x,gy)=0,\text{ for all }x,y\in \C^{2m+\iota|2n}\rangle.
\]
In other words, the Lie superalgebra $\osp_{2m+\iota|2n}$ is the subalgebra of the Lie superalgebra $\gl_{2m+\iota|2n}$ spanned by homogeneous linear operators preserving an even nondegenerate supersymmetric bilinear form. The Lie superalgebra $\osp_{2m+\iota|2n}$ is called an \emph{orthosymplectic Lie superalgebra}. 

One has 
$$
\big(\osp_{2m+\iota|2n}\big)_{\bar 0}\cong \so_{2m+\iota}\oplus \fksp_{2n},\qquad \big(\osp_{2m+\iota|2n}\big)_{\bar 1}\cong \C^{2m+\iota}\otimes \C^{2n},
$$ 
where $\C^{2m+\iota}$ and $\C^{2n}$ 
are vector representations of $\so_{2m+\iota}$ and $\fksp_{2n}$, respectively.

The rank of $\osp_{2m+\iota|2n}$ is $r$ and \[
\dim\osp_{2m+\iota|2n}=\begin{cases} 2r^2+r+2n,\quad &\text{ if }\iota=1;\\
2r^2-r+2n,\quad &\text{ if }\iota=0.
\end{cases}
\]
Denote by $\mathrm{U}(\osp_{2m+\iota|2n})$ the universal enveloping algebra of $\osp_{2m+\iota|2n}$.

Choose the even nondegenerate supersymmetric bilinear form  $(\cdot,\cdot)$ on $\C^{2m+\iota|2n}$ with the Gram matrix 
\beq\label{eq:gram-matrix}
\mathscr G_{2m+1|2n}=\begin{pmatrix}
0 & \mathrm{I}_{m} & 0 & 0 &0\\
\mathrm{I}_{m} & 0 & 0 & 0 &0\\
0 & 0 & 1 & 0 &0\\
0 & 0 & 0 & 0 &\mathrm{I}_n\\
0 & 0 & 0 & -\mathrm{I}_n &0
\end{pmatrix},\qquad
\mathscr G_{2m|2n}=\begin{pmatrix}
0 & \mathrm{I}_{m}  & 0 &0\\
\mathrm{I}_{m} & 0  & 0 &0\\
0 & 0 & 0 &\mathrm{I}_n\\
0 & 0 & -\mathrm{I}_n &0
\end{pmatrix}
\eeq
relative to the standard basis of $\C^{2m+\iota|2n}$. 

For $1\lle i\lle m$ and $1\lle k\lle n$, define
\beq\label{eq:cartan-h}
h_i=E_{ii}-E_{i+m,i+m},\quad  h_{m+k}=E_{k^\circ k^\circ}-E_{k^\circ+n,k^\circ+n}.
\eeq
The \emph{Cartan subalgebra} $\h$ of $\osp_{2m+\iota|2n}$ is the subalgebra spanned by $h_i$ for $i\in I$. 

Consider the dual space $\h^*$ of $\h$. Let $\ve_i$ for $i\in I$ be a basis of $\h^*$ dual to the basis $\{h_i\}_{i\in I}$ of $\h$.
In other words, $\ve_i(h_j)=\delta_{ij}$. 
 We use the convenient notation $\delta_i=\ve_{m+i}$ for $1\lle i\lle n$.

Define a bilinear form $(\cdot, \cdot)$ on $\h^*$ by $(\ve_i,\ve_j)=(-1)^{|i|}\delta_{ij}$.

\subsection{Root systems}In this section, we describe the Cartan matrix for $\osp_{2m+\iota|2n}$ and the root system associated with an arbitrary parity sequence, see e.g. \cite{FSS89}. Throughout the paper, we use the convenient notation $\ve_i^{\bm s}:=\ve_{\sigma_{\bm s}(i)}$.
\subsubsection{The case of $\osp_{2m+1|2n}$} 
Given a parity sequence $\bm s\in S_{m|n}$, we define the set of $\bm s$-\emph{positive roots} 
$$
\Phi_s^+=\{\ve_i^{\bm s}\pm \ve_j^{\bm s}~|~i,j\in I \text{ and }i<j\}\cup \{\ve_i~|~i\in I\}\cup \{2\ve_{i}~|~m+1\lle i\lle r\}.
$$
We have a Dynkin diagram of type B. Define the $\bm s$-\emph{simple positive roots} by
$$
\alpha_i^{\bm s}=\ve_i^{\bm s}-\ve_{i+1}^{\bm s},\qquad i\in I,
$$
where $\ve_{r+1}^{\bm s}=0$. Denote by $\Delta_{\bm s}^+$ the set of $\bm s$-simple positive roots $\{\alpha_i^{\bm s}\}_{i\in I}$.

The \emph{symmetrized Cartan matrix} $B^{\bm s}:=(b_{ij}^{\bm s})_{i,j\in I}=\big((\alpha_i^{\bm s},\alpha_j^{\bm s})\big)_{i,j\in I}$ for $\osp_{2m+1|2n}$ associated with the parity sequence $\bm s$ is an $r\times r$ matrix described by the $2\times 2$ submatrices \eqref{eq:cartan-block} for $1\lle i\lle r-1$ with the convention that $s_{r+1}=0$. The \emph{Cartan matrix} $C^{\bm s}:=(c_{ij}^{\bm s})_{i,j\in I}$ for $\osp_{2m+1|2n}$ associated with the parity sequence $\bm s$ is given by $$c_{ij}^{\bm s}=s_{i}(1+\delta_{ir})b_{ij}^{\bm s}.$$

The Dynkin diagram for $\osp_{2m|2n+1}$ associated with the parity sequence $\bm s$ is obtained from the Dynkin diagram for $\gl_{m|n}$ associated with the parity sequence $\bm s$ by attaching one more node as follows,
\begin{align*}
    \dynkin[root radius = 0.1cm,edge length=1cm,labels={1,2,r-1,r}] B{xx.xo},\qquad & \text{ if }s_{r}=1; \text{ or }\\
    \dynkin[root radius = 0.1cm,edge length=1cm,labels={1,2,r-1,r}] B{xx.x*},\qquad & \text{ if }s_{r}=-1. 
\end{align*}
Here the crosses picture the Dynkin diagram for $\gl_{m|n}$ associated with $\bm s$ and the new root is shown by the circle (empty or filled).

Define the set of $\bm s$-\emph{negative roots} $\Phi_{\bm s}^-:=-\Phi_{\bm s}^+$. Clearly, we have $\Phi=\Phi_{\bm s}^+\cup \Phi_{\bm s}^-$.  

There are two most commonly used root systems for $\osp_{2m+1|2n}$. The \emph{standard positive root system} is the one corresponding to the parity sequence $\bm s_{-}$. In this case, we have the \emph{standard Dynkin diagram}
\[
\dynkin[root radius = 0.1cm,edge length=1cm,labels={1,2,n,n+1,r-1,r}] B{oo.to.oo}
\]
and 
\[
\Delta_{\bm s_{-}}=\{\delta_{i}-\delta_{i+1},~\delta_{n}-\ve_1,~\ve_j-\ve_{j+1},~\ve_m~|~1\lle i\lle n-1,~1\lle j\lle m-1\}.
\] 
The \emph{distinguished positive root system} is the one corresponding to the parity sequence $\bm s_{+}$. In this case, we have the \emph{distinguished Dynkin diagram}
\[
\dynkin[root radius = 0.1cm,edge length=1cm,labels={1,2,m,m+1,r-1,r}] B{oo.to.o*}
\]
and 
\[
\Delta_{\bm s_{+}}=\{\ve_i-\ve_{i+1},~\ve_m-\delta_{1},~\delta_{j}-\delta_{j+1},~\delta_{n}~|~1\lle i\lle m-1,~1\lle j\lle n-1\}.
\]

\subsubsection{The case of $\osp_{2m|2n}$} 
Let us describe the root system associated with a parity sequence $\bm s\in S_{m|n}$. We have two cases corresponding Dynkin diagrams of type C or type D. While in type D a parity sequence determines a unique root system, the new feature is that in type C there are two root systems corresponding to $\bs s$. In order to deal with this phenomenon, we introduce the notation $\tilde{\bs s}$ for each $\bm s\in S_{m|n}$ which includes the parity sequence $\bm s$ and a binary choice $\kappa$:
\begin{align}\label{extended s}
\tilde{\bs s}=(s_1,\dots,s_{r};\kappa), \quad \text{where}\quad \bm s=(s_1,\dots, s_r)\in S_{m|n}, \quad\ \kappa\in \{-1,1\}.
\end{align}
We call $\tl{\bm s}$ an \emph{extended parity sequence}. We use $\kappa(\tl{\bm s})$ to denote the binary choice from $\tl{\bm s}$. Denote $\wt S_{m|n}^0:=S_{m|n}\times \{-1,1\}$ the set of all extended parity sequences.

For convenience, in the case of $\osp_{2m+1|2n}$ we also define the extended parity sequences by mandating $\kappa=1$, that is $\tl{\bm s}=(\bm s;1)$.
Then we denote the set of extended parity sequences in this case by $\wt S_{m|n}^1:=S_{m|n}\times \{1\}$.

\medskip 

Set $\ve_i^{\tl{\bm s}}:=\ve_i^{\bm s}$ if $\ve_i^{\bm s}\ne\ve_m$ and $\ve_i^{\tl{\bm s}}:=\kappa \ve_i^{\bm s}$ if $\ve_i^{\bm s}=\ve_m$. 

\begin{itemize}
    \item If $s_{r}=1$, we have a Dynkin diagram of type D. In this case, we say that $\bm s$ is of type D. Define the $\tl{\bm s}$-\emph{simple positive roots} $$\Delta_{\tl{\bm s}}=\{\alpha_i^{\tl{\bm s}}=\ve_{i}^{\tl{\bm s}}-\ve_{i+1}^{\tl{\bm s}},\quad \alpha_{r}^{\tl{\bm s}}=\ve_{r-1}^{\tl{\bm s}}+\ve_{r}^{\tl{\bm s}},\quad 1\lle i\lle r-1\}.$$As a set $\Delta_{\tl{\bm s}}$ does not depend on $\kappa$. 
    \item If $s_{r}=-1$, we have a Dynkin diagram of type C. In this case, we say that $\bm s$ is of type C. Define the of $\tl{\bm s}$-\emph{simple positive roots} 
    \begin{align*}
    \Delta_{\tl{\bm s}}=\{\alpha_i^{\tl{\bm s}}=\ve_{i}^{\tl{\bm s}}-\ve_{i+1}^{\tl{\bm s}},\quad \alpha_{r}^{\tl{\bm s}}=2\ve_r^{\tl{\bm s}},\quad 1\lle i\lle r-1\}.
    \end{align*}
\end{itemize}
We also use the convention $\Delta_{{\bm s}}^+:=\Delta_{\tl{\bm s}}$ if $\kappa(\tl{\bm s})=1$ and $\Delta_{{\bm s}}^-:=\Delta_{\tl{\bm s}}$ if $\kappa(\tl{\bm s})=-1$.

Define the set of $\tl{\bm s}$-\emph{positive roots} $$
          \Phi_{\tl{\bm s}}^+=\{\ve_{i}^{\tl{\bm s}}\pm \ve_{j}^{\tl{\bm s}}~|~i,j\in I \text{ and }i<j\} \cup \{2\ve_{i}~|~m+1\lle i\lle r\}.
          $$

As before denote $\Phi_{\tl{\bm s}}^-=-\Phi_{\tl{\bm s}}^+$.

The Cartan data and Dynkin diagrams do not depend on the choice of $\kappa$.
The \emph{symmetrized Cartan matrix} $B^{\bm s}:=(b_{ij}^{\bm s})_{i,j\in I}=\big((\alpha_i^{\bm s},\alpha_j^{\bm s})\big)_{i,j\in I}$ for $\osp_{2m|2n}$ associated with the parity sequence $\bm s\in S_{m|n}$ is an $r\times r$ matrix described the $2\times 2$ submatrices \eqref{eq:cartan-block} except that the right-bottom corner $3\times 3$ submatrix is given by 
    $$\begin{matrix}
    (-1,-1) & (1,-1) \\
    \begin{pmatrix}
    s_{r-2}-1 & 1 & 0\\
    1 & -2 & 2\\
    0 & 2 & -4
    \end{pmatrix}, & \begin{pmatrix}
    s_{r-2}+1 & -1 & 0\\
    -1 & 0 & 2\\
    0 & 2 & -4
    \end{pmatrix},
    \end{matrix}$$
    $$\begin{matrix}
    (1,1) & (-1,1) \\
    \begin{pmatrix}
    s_{r-2}+1 & -1 & -1\\
    -1 & 2 & 0\\
    -1 & 0 & 2
    \end{pmatrix}, & \begin{pmatrix}
    s_{r-2}-1 & 1 & 1\\
    1 & 0 & -2\\
    1 & -2 & 0
    \end{pmatrix},
    \end{matrix}$$
where each case corresponds to the parity subsequence $(s_{r-1},s_{r})$ written above it. The \emph{Cartan matrix} $C^{\bm s}:=(c_{ij}^{\bm s})_{i,j\in I}$ for $\osp_{2m|2n}$ associated with the parity sequence $\bm s$ is given by 
$$ 
c_{ij}^{\bm s}=\begin{cases}s_{i}b_{ij}^{\bm s}, &\text{ if } i\neq r,\\
s_{r-1}b_{ij}^{\bm s}, &\text{ if } i=r \text{ and } s_r=1,\\
s_rb_{ij}^{\bm s}/2, &\text{ if } i=r \text{ and } s_r=-1.
\end{cases}
$$

The Dynkin diagram for $\osp_{2m|2n}$ associated with the parity sequence $\bm s$ is obtained from the Dynkin diagram for $\gl_{m|n}$ associated with the parity sequence $\bm s$ by attaching one more node as follows,
\begin{align*}
    \dynkin[root radius = 0.1cm,edge length=1cm,labels={1,2,r-1,r}] C{xx.xo},\qquad & \text{ if }s_{r}=-1; \\
    \dynkin[root radius = 0.1cm,edge length=1cm,labels={1,2,r-3~~~~~~~,r-2,r-1,r}] D{xx.xxoo},\qquad & \text{ if }s_{r-1}=s_{r}=1;\\
    \begin{tikzpicture}[baseline]
\dynkin[root radius = 0.1cm,edge length=1cm,labels={1,2,r-3~~~~~~~,r-2,r-1,r}] D{xx.xxtt}
\draw[double,double distance=2pt] (root 5) to [out=-90, in=90] (root 6);
\end{tikzpicture},\qquad & \text{ if }s_{r-1}=-1,~s_{r}=1.
\end{align*}

Here the crosses and the $(r-1)$-st node picture the Dynkin diagram for $\gl_{m|n}$ associated with $\bm s$ and the new root corresponds to the $r$-th node shown by a circle (empty or with the cross).


Again, there are two most commonly used root systems for $\osp_{2m|2n}$. The \emph{standard positive root system} is the one corresponding to the extended parity sequence $\tlbms_{-}:=(\bm s_-;1)$. In this case, we have the \emph{standard Dynkin diagram}
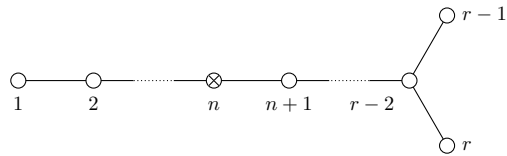
\[
\dynkin[root radius = 0.1cm,edge length=1cm,labels={1,2,n,n+1,r-2,r-1,r}] D{oo.to.ooo}
\]
and 
\[
\Delta_{\bm s_-}^+=\{\delta_{i}-\delta_{i+1},~\delta_{n}-\ve_1,~\ve_j-\ve_{j+1},~\ve_{m-1}+\ve_{m}~|~1\lle i\lle n-1,~1\lle j\lle m-1\}.
\]
The \emph{distinguished positive root system} is the one corresponding to the extended parity sequence $\tlbms_{+}:=(\bm s_+;1)$. In this case, we have the \emph{distinguished Dynkin diagram}
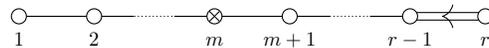
\[
\dynkin[root radius = 0.1cm,edge length=1cm,labels={1,2,m,m+1,r-1,r}] C{oo.to.oo}
\]
and 
\[
\Delta_{\bm s_+}^+=\{\ve_i-\ve_{i+1},~\ve_m-\delta_{1},~\delta_{j}-\delta_{j+1},~2\delta_{n}~|~1\lle i\lle m-1,~1\lle j\lle n-1\}.
\]

\subsection{Borel subalgebras}
In this section, we recall the root vectors of $\osp_{2m+\iota|2n}$.

There is a basis of $\osp_{2m+1|2n}$ given as follows:
\begin{align}
    h_i=E_{ii}-E_{i+m,i+m},\quad & h_{m+k}=E_{k^\circ k^\circ}-E_{k^\circ+n,k^\circ+n},\nonumber\\
    e_{\ve_i}=E_{2m+1,i+m}-E_{i,2m+1},\quad & e_{-\ve_i}=E_{i+m,2m+1}-E_{2m+1,i}, \label{eq:skip1} \\
    e_{2\delta_k}=\sqrt{2}E_{k^\circ,k^\circ+n},\quad & e_{2\delta_k}=-\sqrt{2}E_{k^\circ+n,k^\circ}\nonumber\\
    e_{\delta_k+\delta_l}=E_{k^\circ,l^\circ+n}+E_{l^\circ,k^\circ+n},\quad & e_{-\delta_k-\delta_l}=-E_{l^\circ+n,k^\circ}-E_{k^\circ+n,l^\circ},\nonumber\\
    e_{\delta_k-\delta_l}=\mathsf{sgn}(k-l)(E_{k^\circ l^\circ}-E_{l^\circ+n,k^\circ+n}),\quad & e_{\ve_i-\ve_j}=E_{ij}-E_{j+m,i+m},\nonumber\\
    e_{\ve_i+\ve_j}=E_{i,j+m}-E_{j,i+m},\quad & e_{-\ve_i-\ve_j}=E_{j+m,i}-E_{i+m,j},\nonumber\\
    e_{\delta_k+\ve_i}=E_{i,k^\circ+n}+E_{k^\circ,i+m},\quad & e_{-\delta_k-\ve_i}=E_{i+m,k^\circ}-E_{k^\circ+n,i},\nonumber\\
    e_{\delta_k-\ve_i}=E_{i+m,k^\circ+n}+E_{k^\circ,i},\quad & e_{-\delta_k+\ve_i}=E_{i,k^\circ}-E_{k^\circ+n,i+m},\nonumber\\
    e_{\delta_k}=E_{2m+1,k^\circ+n}+E_{k^\circ,2m+1},\quad & e_{-\delta_k}=E_{2m+1,k^\circ}-E_{k^\circ+n,2m+1},\label{eq:skip2} 
\end{align}
for $1\lle i\ne j\lle m$ and $1\lle k\ne l\lle n$, where $\mathsf{sgn}(x)$ denotes the sign of $x$.

Similarly, a basis of $\osp_{2m|2n}$ is given by the same formulas,  for $1\lle i\ne j\lle m$ and $1\lle k\ne l\lle n$, with equations \eqref{eq:skip1} and \eqref{eq:skip2} skipped.

The (quadratic) \emph{Casimir element} $\Omega$ is given by
\[
\Omega:=\frac{1}{4}\Big(\sum_{i=1}^m h_i^2-\sum_{k=1}^n h_{m+k}^2+\sum_{\alpha\in\Phi}e_{\alpha}e_{-\alpha}\Big)\in \mathrm{U}(\osp_{2m+\iota|2n}).
\]
The Casimir element $\Omega$ is central in $\mathrm{U}(\osp_{2m+\iota|2n})$.

The \emph{nilpotent subalgebra $\n_{\tl{\bm s}}^+$} of $\osp_{2m+\iota|2n}$ associated with $\tl{\bm s}$ is generated by $\{e_{\alpha}~|~\alpha\in \Delta_{\tl{\bm s}}\}$. A basis of the nilpotent algebra $\n_{\tl{\bm s}}^+$ is given by $\{e_{\alpha}~|~\alpha\in \Phi_{\tl{\bm s}}^+\}$. The \emph{Borel subalgebra associated with} $\tl{\bm s}$ is $\fkb_{\tl{\bm s}}^+:=\h\oplus \n_{\tl{\bm s}}^+$. The Borel subalgebra with $\tl{\bm s}_-$ is the \emph{standard Borel subalgebra}.

\subsection{Representations}
Let $V$ be an $\osp_{2m+\iota|2n}$-module. Let $\tl{\bm s}\in \wt S_{m|n}^\iota$ be an extended parity sequence.  

For a weight $\lambda\in\h^*$, we call a nonzero vector $v_{\la}^{\tl{\bm s}}\in V$ an $\tl{\bm s}$-\emph{singular vector} of weight $\la$ if $\n_{\tl{\bm s}}^+v_{\la}^{\tl{\bm s}}=0$ and $hv_{\la}^{\tl{\bm s}}=\la(h)v_{\la}^{\tl{\bm s}}$, for all $h\in \h$. Denote by $V^{\bsing}$ the subspace of $\tl{\bm s}$-singular vectors in $V$.  Set
\[
V_{\la}:=\{v\in V~|~hv=\la(h)v\},\qquad V_{\la}^\bsing:=V^\bsing \cap V_{\la}.
\]
Denote by $L^{\tl{\bm s}}(\la)$ the irreducible module generated by an $\tl{\bm s}$-singular vector $v_{\la}^{\bsing}$ of weight $\la$.  We simply write $L(\la)$ if $\tlbms$ corresponds to the standard root system.

In this paper we study Bethe ansatz in the tensor products of finite-dimensional $\osp_{2m+\iota|2n}$-modules given in the parity sequence $\bm s_-$ by hook partitions as follows. 
Let $\mu=(\mu_1\gge \mu_2\gge \cdots)$ be a partition: $\mu_i\in \Z_{\gge 0}$ and $\mu_j=0$ for sufficiently large $j$. We use $\mu'=(\mu_1' \gge \mu_2' \gge \cdots)$ to denote the conjugate of $\mu$. We call a partition $\mu$ an $(m|n)$-\emph{hook partition} if $\mu_{m+1}\lle n$. 

For an $(n|m)$-hook partition $\mu$, define two $\osp_{2m+\iota|2n}$-weights $\mu_\pm$  by
\beq\label{eq:mu-pm}
\begin{split}
&\mu_+=\sum_{i=1}^n \mu_i\delta_i+\sum_{j=1}^m \max\{\mu_j'-n,0\}\epsilon_{j}, \\
&\mu_-=\sum_{i=1}^n \mu_i\delta_i+\sum_{j=1}^{m-1} \max\{\mu_j'-n,0\}\epsilon_{j}-\max\{\mu_m'-n,0\}\epsilon_{m}.
\end{split}
\eeq
We call such $\mu_+$ (resp. $\mu_\pm$) \emph{dominant integral} $\osp_{2m+1|2n}$-weights (resp. $\osp_{2m|2n}$-weights). Note that $L(\mu_+)$ ($L(\mu_\pm)$ for $\iota=0$ case) are finite-dimensional $\osp_{2m+\iota|2n}$-modules. We call $\mu^+$ and $L(\mu^+)$ (resp.  $\mu^-$ and $L(\mu^-)$) the dominant integral weight and the highest weight module of the first kind (resp. of the second kind).

\medskip

Let $\la$ be a dominant integral  $\osp_{2m+\iota|2n}$-weight. Then $L(\la)$ is a highest weight module with respect to all Borel subalgebras. The highest weights and the highest weight vectors depend on the choice of the Borel subalgebras - that is on the extended parity sequences $\tlbms$. We now describe this dependence.

Denote by $\la^{\tlbms}$ the unique $\osp_{2m+\iota|2n}$-weight such that $L^{\tlbms}(\la^{\tlbms})\cong L(\la)$. Note that if $\tlbms$ is of type D, then $\la^{\tlbms}$ is independent of $\kappa$. 

For $\tlbms \in\wt S_{m|n}^\iota$ and an $\osp_{2m+\iota|2n}$-weight $\la$, denote by
\[
\la_{[\tlbms]}=(\la_{[\tlbms],1},\dots,\la_{[\tlbms],r}),\qquad \text{where }\qquad \la_{[\tlbms],i}=s_i(\la^{\tlbms},\ve_i^{\tlbms}),
\]
the coordinate sequence of $\la$ associated to $\tlbms$. The coordinate sequence of $\la^{\tlbms}$ can be computed recursively as follows.

Let $1\lle i\lle r-1$. For $\bm s \in S_{m|n}$, let $\bm s^{[i]}=(s_1,\dots,s_{i+1},s_i,\dots,s_r)\in S_{m|n}$ be the parity sequence obtained from $\bm s$ by interchanging $i$-th and $(i + 1)$-st components. Let $\tlbms \in \wt S_{m|n}^\iota$. 

Set
\beq\label{eq:tlbms}
\begin{split}
&\tlbms^{[i]}=(\bm s^{[i]};\kappa(\tlbms))\in \wt S_{m|n}^{\iota},\qquad 1\lle i\lle r-1,\\
&\tlbms^{[r]}=\begin{cases} \tlbms,\quad &\text{ if }\bm s\text{ is of type C},\\ (\bm s^{[r-1]};s_{r-1}\kappa(\tlbms)),\quad &\text{ if }\bm s\text{ is of type D}.\end{cases}
\end{split}
\eeq
In addition if $\iota=0$ and $s_r=1$ (the case of type D), we also set $\tlbms^{[f]}=(\bm s;-\kappa(\tlbms))\in \wt S_{m|n}^0$. Here $[f]$ stands for ``fake". The reason for this will become clear later, see Lemma \ref{lem:fake}. 
Then the change in the $r$-th direction is a composition: 
\beq\label{eq:r-change}
\tlbms^{[r]}=\begin{cases}
((\tlbms^{[f]})^{[r-1]})^{[f]},& \text{ if }s_{r-1}=1,\\
(\tlbms^{[f]})^{[r-1]}, & \text{ if }s_{r-1}=-1.\end{cases}
\eeq
We also set $\bm s^{[r]}=\bm s^{[r-1]}$, $\bm s^{[f]}=\bm s$.

\begin{lem}\label{lem: weight change}
Let $\la$ be a dominant integral $\osp_{2m+\iota|2n}$-weight, and $\tlbms \in \wt S_{m|n}^\iota$ an extended parity sequence. 

For $1\lle i\lle r-1$, if $s_i\ne s_{i+1}$, then
\[
\la_{[\tlbms^{[i]}]}^{\tlbms^{[i]}}=(\la_{[\tlbms],1}^{\tlbms},\dots,\la_{[\tlbms],i-1}^{\tlbms},\la_{[\tlbms],i+1}^{\tlbms}+\eta,\la_{[\tlbms],i}^{\tlbms}-\eta,\dots,\la_{[\tlbms],r}^{\tlbms}),
\]
where $\eta=1$ if $\la_{[\tlbms],i}^{\tlbms}+\la_{[\tlbms],i+1}^{\tlbms}\ne 0$ and $\eta=0$ otherwise. 

If $\iota=0$ and $s_r=1$ (that is if in the case of type D), then we have
\[
\la_{[\tlbms^{[f]}]}^{\tlbms^{[f]}}=(\la_{[\tlbms],1}^{\tlbms},\dots,\la_{[\tlbms],r-1}^{\tlbms},-\la_{[\tlbms],r}^{\tlbms}).
\]
\end{lem}
\begin{proof}
Let $v^{\tlbms}$ be an $\tlbms$-singular vector of $L(\la)$. If $\la_{[\tlbms],i}^{\tlbms}+\la_{[\tlbms],i+1}^{\tlbms}=0$ then $v^{\tlbms}$ is also an $\tlbms^{[i]}$-singular vector. Otherwise, let $e$ be a root vector of $\osp_{2m+\iota|2n}$ of weight $-\alpha_i^{\tlbms}=\ve_{i+1}^{\tlbms}-\ve_i^{\tlbms}$. Then $ev^{\tlbms}$ is a non-zero $\tlbms^{[i]}$-singular vector. 

Since
$\ve_j^{\tlbms^{[i]}}=\ve_j^{\tlbms}$, $j\neq i,i+1$ and $\ve_i^{\tlbms^{[i]}}=\ve_{i+1}^{\tlbms}$,
$\ve_{i+1}^{\tlbms^{[i]}}=\ve_{i}^{\tlbms}$, the lemma follows.
\end{proof}
For our purposes, we do not need to solve this recursion, see \cite[Section 2.4]{CW12} for a partial answer. 

\section{Symmetric Rational pseudo-differential operators and isotropic superflags}\label{operators sec}
\subsection{Rational pseudo-differential operators}We recall the basics of rational pseudo-differential operators from \cite{CDSK12,HMVY19}.

Let $\mathcal K=\C(x)$ be the differential field of complex-valued rational functions with the derivation $\pa$. Consider the division ring of pseudo-differential operators $\mc K((\pa^{-1}))$. An element in $\mc K((\pa^{-1}))$ is of the form
\[
A=\sum_{j=-\infty}^M a_j\pa^j,\qquad a_j\in\mc K,\qquad M\in\Z.
\]We say that the \emph{order} of $A$ is $M$ if $a_M\ne 0$, and denote it by $\mathrm{ord}~A$. We say that $A$ is \emph{monic} if $a_M=1$.

For $j\in \Z_{\gge 0}$ and $N\in \Z$, set
\[
{N \choose j}=\frac{N(N-1)\dots(N-j+1)}{j!}.
\]
One has the following relations in $\mc K((\pa^{-1}))$:
\[
\pa\pa^{-1}=\pa^{-1}\pa=1,\quad \pa^N a=\sum_{j=0}^\infty {N \choose j}a^{(j)}\pa^{N-j},\quad a\in\mc K,\quad N\in \Z,
\]
where $a^{(j)}$ is the $j$-th derivative of $a$ and $a^{(0)}=a$. Note that all nonzero elements in $\mc K((\pa^{-1}))$ are invertible.

Consider the algebra of differential operators $\mc K[\pa]$ which is a subring of $\mc K((\pa^{-1}))$. Let $\cD\in\cK[\pa]$ be a monic differential operator of order $M$. We say that $\cD$ is \emph{complete factorable} over $\cK$ if $\cD=d_1\cdots d_M$, where $d_i=\pa-a_i$ for some $a_i\in\cK$ and each $1\lle i\lle M$. Denote by $\ker\cD$ the set $\{f~|~\cD f=0\}$. If $\dim(\ker\cD)=\ord\cD$ and $\ker\cD\subset \mc K$, then $\cD$ is completely factorable over $\cK$.

Denote by $\mc K(\pa)$ the division subring of $\mc K((\pa^{-1}))$ generated by $\mc K[\pa]$. We call an element in $\mc K(\pa)$ a \emph{rational pseudo-differential operator}. Let $\cR$ be a rational pseudo-differential operator. If $\cR$ is of the form $\cD_{\bar 0}\cD_{\bar 1}^{-1}$ for some $\mc D_{\bar 0},\cD_{\bar 1}\in\cK[\pa]$, then we call $\cD_{\bar 0}\cD_{\bar 1}^{-1}$ a \emph{fractional factorization} of $\cR$. We say that a fractional factorization $\cR=\cD_{\bar 0}\cD_{\bar 1}^{-1}$ is \emph{minimal} if $\cD_{\bar 1}$ is monic and of minimal possible order.

\begin{prop}[\cite{CDSK12}]
Let $\cR\in\cK(\pa)$, then there exists a unique minimal fractional factorization of $\cR$. Let $\cR=\cD_{\bar 0}\cD_{\bar 1}^{-1}$ be such that $\dim(\ker\cD_{\bar 0})=\ord\cD_{\bar 0}$ and $\dim(\ker\cD_{\bar 1})=\ord\cD_{\bar 1}$. Then $\cR=\cD_{\bar 0}\cD_{\bar 1}^{-1}$ is the minimal fractional factorization of $\cR$ if and only if $\ker\cD_{\bar 0}\cap \ker\cD_{\bar 1}=0$.\qed
\end{prop}

We call $\cR$ an \emph{$(m|n)$-rational pseudo-differential operator} if $\mc R$ is monic and for the minimal fractional factorization $\cR=\cD_{\bar 0}\cD_{\bar 1}^{-1}$, we have that $\cD_{\bar 0}$ and $\cD_{\bar 1}$ are completely factorable over $\cK$, and $\ord\cD_{\bar 0}=m$ and $\ord\cD_{\bar 1}=n$.

Let $\cR$ be an $(m|n)$-rational pseudo-differential operator. Let $\bm s\in S_{m|n}$. We call the form $\cR=d_1^{s_1}\cdots d_{r}^{s_r}$, where $d_i=\pa-a_i$, $a_i\in \cK$, $i\in I$, a \emph{complete $\bm s$-factorization}. Denote by $\mc F^{\bm s}(\cR)$ the set of all complete $\bm s$-factorizations of $\cR$. We call $\mc F(\cR)=\bigsqcup_{\bm s\in S_{m|n}}\mc F^{\bm s}(\cR)$ the \emph{set of all complete factorizations}.

Suppose $\cR_1=(\pa-a)(\pa-b)^{-1}$ and $\cR_2=(\pa-c)^{-1}(\pa-d)$ be two $(1|1)$-rational pseudo-differential operators, where $a,b,c,d\in\cK$, $a\ne b$, and $c\ne d$. Then $\cR_1=\cR_2$ if and only if 
\beq\label{eq:switch}
\begin{cases}
c=b+\ln'(a-b),\\
d=a+\ln'(a-b),
\end{cases}
\iff \quad 
\begin{cases}
a=d-\ln'(c-d),\\
b=c-\ln'(c-d),
\end{cases}
\eeq
where $\ln'(f)=f'/f$.

Let $\cR$ be an $(m|n)$-rational pseudo-differential operator. Suppose $\cR=d_1^{s_1}\cdots d_r^{s_r}$, $d_i=\pa-a_i$, is a complete $\bm s$-factorization. If $s_i\ne s_{i+1}$, then $d_i\ne d_{i+1}$. Using \eqref{eq:switch}, we obtain $\tl d_i$ and $\tl d_{i+1}$ such that $d_i^{s_i}d_{i+1}^{s_{i+1}}=\tl d_i^{s_{i+1}}\tl d_{i+1}^{s_i}$, which gives a complete $\bm s^{[i]}$-factorization $\cR=d_1^{s_1}\cdots\tl d_i^{s_{i+1}}\tl d_{i+1}^{s_i}\cdots d_r^{s_r}$ with $\bm s^{[i]}=(s_1,\dots,s_{i+1},s_i,\dots,s_r)$.  Repeating this procedure, we see that there exists a canonical bijection between the sets of complete
factorizations with respect to any two parity sequences.

\subsection{Factorizations and superflag varieties}
Let $W=W_{\bar 0}\oplus W_{\bar 1}$ be a vector superspace with $\dim(W_{\bar 0})=m$ and $\dim(W_{\bar 1})=n$. Consider a \emph{full flag} $\sF$ of $W$, $\sF=\{F_1\subset F_2\subset \dots\subset F_{r}=W\}$ such that $\dim(F_i)=i$. We say that a basis $\{w_1,\dots,w_{r}\}$ of $W$ \emph{generates the full flag $\sF$} if $F_i$ is spanned by $w_1,\dots,w_i$. A full flag is called a \emph{full superflag} if it is generated by a homogeneous basis. Denote by $\sF(W)$ the set of all full superflags.
We note that $\dim \sF(W)=m(m-1)/2+n(n-1)/2$.

To a homogeneous basis $\{w_1,\dots,w_{r}\}$ of $W$, we associate the unique parity sequence $\bm s\in S_{m|n}$ such that $s_i=(-1)^{|w_i|}$. We say a full superflag $\sF$ \emph{has parity sequence} $\bm s$ if it is generated by a homogeneous basis whose parity sequence is $\bm s$. We denote by $\sF^{\bm s}(W)$ the set of all full superflags of parity $\bm s$.

Clearly, we have
\[
\sF(W)=	\bigsqcup_{\bm{s}\in S_{m|n}} \sF^{\bm s}(W), \qquad 
\sF^{\bm s}(W)\cong\sF\left(W_{\bar{0}}\right)\times \sF\left(W_{\bar{1}}\right).
\] 

Given a basis $\{v_1,\dots, v_m\}$ of $W_{\bar 0}$, a basis $\{u_1,\dots, u_n\}$ of $W_{\bar 1}$, and a parity sequence $\bm s\in S_{m|n}$, define a homogeneous basis $\{w_1,\dots, w_{r}\}$ of $W$ by the rule $w_i=v_{\bm s^+_i+1}$ if $s_i=1$ and 
$w_i=u_{\bm s^-_i+1}$ if $s_i=-1$. Conversely, any homogeneous basis of $W$ gives a basis of $W_{\bar 0}$, a basis of $W_{\bar 1}$, and a parity sequence $\bm s$. We say that the basis $\{w_1,\dots, w_{r}\}$ is \emph{associated to $\{v_1,\dots, v_m\}$, $\{u_1,\dots, u_n\}$, and $\bm s$}.

Define the \emph{Wronskian} $\Wr$ of $g_1,\dots,g_k$ by
$$
\Wr(g_1,\dots,g_k)=\det \left( g_j^{(i-1)} \right)_{i,j=1}^k.
$$
By convention, we extend this definition to the case of $k=0$ by setting $\Wr(\emptyset)=1$.

Let $\cR$ be an $(m|n)$-rational pseudo-differential operator. Let $\cR=\cD_{\bar 0}\cD_{\bar 1}^{-1}$ be the minimal fractional factorization. Let $V=W_{\bar 0}=\ker \cD_{\bar 0} $, $U=W_{\bar 1}=\ker \cD_{\bar 1} $, $W=W_{\bar 0}\oplus W_{\bar 1}$. We call $W$ the \emph{superkernel} of $\cR$ and denote it by $\mathrm{sker}\cR$.

Given a basis $\{v_1,\dots, v_m\}$ of $V$, a basis $\{u_1,\dots, u_n\}$ of $U$, and a parity sequence $\bm s\in S_{m|n}$, define $d_i=\pa-f_i$, where
\beq\label{eq wronski coeff}
\begin{split}
&f_i= \ln'  \frac{\Wr(v_1,v_2,\dots,v_{\bm s_i^++1},u_1,u_2,\dots,u_{\bm s_i^-})}{\Wr(v_1,v_2,\dots,v_{\bm s_i^+},u_1,u_2,\dots,u_{\bm s_i^-})}, \qquad {\rm if}\ s_i=1, \\ &f_i= \ln'  \frac{\Wr(v_1,v_2,\dots,v_{\bm s_i^+},u_1,u_2,\dots,u_{\bm s_i^-+1})}{\Wr(v_1,v_2,\dots,v_{\bm s_i^+},u_1,u_2,\dots,u_{\bm s_i^-})},  \qquad {\rm if}\ s_i=-1.
\end{split}
\eeq

Note that if two bases $\{v_1,\dots, v_m\}$, $\{\tl{v}_1,\dots, \tl{v}_m\}$ generate the same full flag of $V$ and two bases $\{u_1,\dots, u_n\}$, $\{\tl{u}_1,\dots, \tl{u}_n\}$ generate the same full flag of $U$, then the coefficients $f_i$ computed from $v_j,u_j$ and from $\tl{v}_j,\tl{u}_j$ are the same.

\begin{prop}[{\cite{HMVY19}}]\label{prop flag factor}
We have a complete $\bm s$-factorization of $\cR:$ $\cR=d_1^{s_1}\cdots d_{r}^{s_{r}}$.\qed
\end{prop}

By Proposition \ref{prop flag factor}, we have maps $\varpi: \sF(W) \to \mc F(\cR)$ and $\varpi^{\bm s}: \mathscr F^{\bm s}(W) \to \mc F^{\bm s}(\cR)$. The map $\varpi^{\bm s}$ is explicitly given as follows. Let $\mathscr F$ be a superflag generated by a basis which is associated to $\{v_1,\dots, v_m\}$, $\{u_1,\dots, u_n\}$, and $\bm s$. Then
\[
\varpi^{\bm s}(\mathscr F)=(\pa-f_1 )^{s_1}(\pa-f_2 )^{s_2}\cdots(\pa-f_{r})^{s_{r}},
\]
where $f_i$ are given by \eqref{eq wronski coeff}. The map $\varpi$ is the disjoint union of maps $\varpi^{\bm s}$ over all distinct parities $\bm s$.
\begin{lem}\label{lem:A-bij-flag-factor}
The maps $\varpi$ and $\varpi^{\bm s}$ are bijections.\qed
\end{lem}
Thus the set of complete factorizations of $\cR$ is canonically identified with the variety of full superflags of $W$. 

\subsection{Symmetric differential operators and self-dual spaces}
We define the anti-involution $*$ on $\cK[\pa]$ by
\[
\Big(\sum_{j=0}^N a_j \pa^j\Big)^* = \sum_{j=0}^N (-\pa)^j a_j.
\]
We call a differential operator $\cD\in\cK[\pa]$ \emph{symmetric} if $\cD^*=\cD$ or if $\cD^*=-\cD$. We do not distinguish differential operators up to a scalar multiple.

Let $V$ be a finite-dimensional vector space of functions. 
We will always assume that the functions are sufficiently smooth.
For example, one can assume the functions are holomorphic in some domain.  In most cases, we will have rational functions.

Denote by $\cD_V$ the monic differential operator whose kernel is $V$. Denote by $V^*$ the kernel of $\cD_V^*$. 

Let $A=(a_{ij})_{i,j=1}^N$ be an $N\times N$ matrix with possibly noncommuting entries. Define the \emph{row determinant} of $A$ by
\[
\mathrm{rdet}\, A:=\sum_{\sigma\in\fkS_N}(-1)^{\sigma}a_{1\sigma(1)}a_{2\sigma(2)}\cdots a_{N\sigma(N)}.
\]

Let $v_1,\dots,v_N$ be a basis of $V$, then we have
\beq\label{eq:DV}
\cD_V=\frac{1}{\Wr(v_1,\dots,v_N)}\mathrm{rdet}\begin{pmatrix} v_1 & v_1' & \dots & v_1^{(N)}\\
v_2 & v_2' & \dots & v_2^{(N)}\\
\vdots & \vdots & \ddots & \vdots \\
v_N & v_N' & \dots & v_N^{(N)}\\
1 & \pa & \dots & \pa^N
\end{pmatrix}.
\eeq

The space $V$ is called \emph{self-dual} if there exists some function $f(x)$ such that
\begin{equation}\label{f}
f^2(x)\cD_V^*=\cD_V f^2(x)\quad  {\text{or}} \quad f^2(x)\cD_V^*=-\cD_V f^2(x). 
\end{equation}
In other words, the space $V$ is self-dual if $V=f^2(x)V^*$ for some function $f(x)$.

The space $V$ is called \emph{normalized self-dual} if $\cD_V$ is symmetric. In other words, $V$ is normalized self-dual if $V=V^*$.

If $V$ is self-dual and $f(x)$ is as in \eqref{f} then $V/f(x)$ is normalized self-dual.

\medskip

A function $f(x)\in\cK$ is called \emph{monic} if it can be written as a ratio of two monic polynomials.

Let $V$ be a space with a basis $v_1,\dots,v_N$. Assume, $\Wr(v_1,\dots,v_N)$ is a rational function. Denote
the unique monic rational function which is a scalar multiple of $\Wr(v_1,\dots,v_N)$ by $\Wr(V)$. 

\begin{lem}\label{lem:dual}
Let $ v_1,\dots,v_N $ be a basis of $V$, then a basis of $V^*$ is given by
\[
\frac{\Wr(v_1,\dots,\widehat{v_i},\dots,v_N)}{\Wr(v_1,\dots,v_i,\dots,v_N)},\qquad 1\lle i\lle N,
\]
where the symbol $\widehat{v_i}$ means that $v_i$ is skipped. 

In particular, if $V$ is normalized self-dual, then $\Wr(V)=1$. If $V$ is self-dual and $f$ as in \eqref{f} is monic then
$f^N=\Wr(V)$.
\end{lem}
\begin{proof}
The first statement follows from \cite[Theorem 3.14]{MTV08} and the second statement follows from the first statement and the standard Wronskian identities, see \cite[Lemmas A.2, A.5]{MV04}.
\end{proof}

Let $V$ be a normalized self-dual space of functions.
By Lemma \ref{lem:dual}, we have a linear isomorphism
\[
\La^{N-1}(V)\to V^*,\qquad v_{1}\wedge \dots\wedge \widehat{v_i}\wedge\dots \wedge v_N\mapsto \Wr(v_{1},\dots,\widehat{v_i},\dots, v_N),
\]
where $\La^{N-1}(V)$ is the $(N-1)$-st exterior power of $V$.

Define a nondegenerate bilinear form
\[
(\cdot,\cdot):\ V\otimes V\to \C, \quad v\otimes w\mapsto\Wr(v,w_1,\dots,w_{N-1}), \quad {\text{if}} \quad w=\Wr(w_1, \dots ,w_{N-1}), \quad w_i\in V.
\]
We call this bilinear form the \emph{canonical bilinear form}.

A basis $\{v_1,\dots,v_N\}$ is called a \emph{Witt basis} if
\beq\label{eq:witt}
v_{N+1-i}=\Wr(v_1,\dots,\widehat{v_i},\dots,v_N),\qquad 1\lle i\lle N.
\eeq
\begin{prop}[{\cite[Theorem 6.4]{MV04}}]
If $V$ is a normalized self-dual space, then $V$ has a Witt basis.
\end{prop}
\begin{proof}
The statement is given in \cite{MV04} for spaces of polynomials. The proof applies for arbitrary spaces of functions if instead of degrees one uses vanishing orders at some $z\in\C$.
\end{proof}

\begin{cor}
If $V$ is a normalized self-dual space, then the canonical form is symmetric if $N$ is odd and skew-symmetric if $N$ is even. \qed
\end{cor}

\begin{lem}\label{normalize}
Let $V$ be a normalized self-dual space. Let $v_1,\dots,v_N\in V$ satisfy $(v_i,v_j)=(-1)^{i-1}\delta_{i+j,N+1}$.
Let $a\in\C^\times$ be such that $a^{N-2}=\Wr(v_1,\dots,v_N)$.

Then $\{v_1/a,v_2/a,\dots,v_N/a\}$ is a Witt basis of $V$.
\end{lem}
\begin{proof}
Consider $w=\Wr(v_1,\dots,\widehat{v_i},\dots,v_N)$. Then 
by the definition of the canonical form $(v_j,w)=0$, $j\neq i$, and $(v_i,w)=(-1)^{i-1}\Wr(v_1,\dots,v_N)=(-1)^{i-1}a^{N-2}$. Since the canonical form is nondegenerate we obtain $w=a^{N-2} v_{N+1-i}$. The lemma follows.
\end{proof}

Given a subspace $\mathscr V$ of $V$, denote by $\mathscr V^\perp$ its orthogonal complement with respect to the canonical bilinear form.

\begin{lem}\label{lem:wr-orho-complement}
If $V$ is a normalized self-dual space, then $\Wr(\mathscr V)=\Wr(\mathscr V^\perp)$ for all $\mathscr V\subset V$.
\end{lem}
\begin{proof}
Consider the case of $\mathscr V=\langle v_{a_1},\dots,v_{a_k}\rangle$ where $1\lle a_1<\dots<a_k\lle N$ and $v_1,\dots,v_N$ is a Witt basis of $V$. Then $\mathscr V^\perp=\langle v_{b_1},\dots,v_{b_{N-k}}\rangle$ where $1\lle b_1<\dots<b_{N-k}\lle N$ and  $a_i+ b_j\ne N+1$ for all $i,j$. Note that $\Wr(V)=1$, the lemma in this case follows directly from \cite[Lemma A.5]{MV04} and \eqref{eq:witt}.

Let now $\mathscr V\subset V$ be any subspace of dimension $k$. We reduce this case to the one described above as follows. Choose a basis $w_1,\dots, w_k$ of $\mathscr V$ such that $(w_i,w_j)=\delta_{i+j,l+1}$ for some $l\lle k$. In particular, $\mathscr V \cap \mathscr V^\perp$ has basis $\{w_{l+1},\dots, w_k\}$. Extend $w_{l+1},\dots, w_k$ to a basis $\{w_{l+1},\dots,w_k,u_1,\dots,u_{N-2k+l}\}$ of $\mathscr V^\perp$ such that $(u_i,u_j)=\delta_{i+j,N-2k+l+1}$, $(w_i,u_j)=0$.

By Lemma \ref{normalize}, there exists a Witt basis $\{v_1,\dots,v_N\}$ of $V$ such that all $w_i$ and $u_i$ are equal to $c_jv_j$ for some constants $c_j$. Thus the lemma follows from the special case considered above.
\end{proof}

\subsection{Symmetric rational pseudo-differential operators and self-dual superspaces}
The anti-involution $*$ on $\cK[\pa]$ induces an anti-involution on $\cK(\pa)$ which we denote again by $*$. 

A monic rational pseudo-differential operator $\cR$ is called \emph{symmetric} if $\cR^*=\pm \cR$. We will not make a distinction between rational pseudo-differential operators $\cR_1$ and $\cR_2$ if $\cR_1=k\cR_2$ for $k\in\C^\times$. Hence if $\cR$ is symmetric, we have $\cR^*=\cR$.

Let $\cR$ be an $(M|N)$-rational pseudo-differential operator with a complete $\bm s$-factorization
\[
\cR=(\pa-f_1 )^{s_1}(\pa-f_2 )^{s_2}\cdots(\pa-f_{M+N-1})^{s_{M+N-1}}(\pa-f_{M+N})^{s_{M+N}},
\]
then 
\[
\cR^*=(\pa+f_{M+N})^{s_{M+N}}(\pa+f_{M+N-1})^{s_{M+N-1}}\cdots(\pa+f_2 )^{s_2}(\pa+f_1 )^{s_1}.
\]
In particular, if $s_i=s_{M+N+1-i}$ and $f_i+f_{M+N+1-i}=0$, then $\cR$ is symmetric.

Let $V$ and $U$ be two spaces of rational functions such that $V\cap U=0$. Let $v_1,\dots,v_M$ and $u_1,\dots,u_N$ be bases of $V$ and $U$, respectively.  Define
\beq\label{eq:V-U}
V_U:=\left\langle \frac{\Wr(u_1,\dots,u_N,v_i)}{\Wr(U)} \right \rangle_{1\lle i\lle M},\qquad U_V:=\left\langle \frac{\Wr(v_1,\dots,v_M,u_j)}{\Wr(V)} \right \rangle_{1\lle j\lle N}.
\eeq

Let $\cR$ be an $(M|N)$-rational pseudo-differential operator with the superkernel $W$. Let $V=W_{\bar 0}$ and $U=W_{\bar 1}$, then $\cR=\cD_V^{}\cD_{U}^{-1}$.

\begin{lem}\label{lem:exchange-pm}
We have $\cR=\cD_{U_V}^{-1}\cD_{V_U}^{}$.
\end{lem}
\begin{proof}
It reduces to show that $\cD_{U_V}\cD_{V}=\cD_{V_U}\cD_{U}$. We claim that $\cD_{V\oplus U}=\cD_{U_V}\cD_{V}$. Indeed, it is clear that $V\subset \ker (\cD_{U_V}\cD_{V})$. Hence we only need to show $U\subset \ker (\cD_{U_V}\cD_{V})$. It follows from \eqref{eq:DV} that
\[
\cD_V(u_i)=\mathrm{const}\cdot\frac{\Wr(v_1,\dots,v_M,u_i)}{\Wr(V)}.
\]
Hence $\cD_V(u_i)\in U_V$. Therefore $u_i\in \ker (\cD_{U_V}\cD_{V})$ and the claim follows. Similarly, we have $\cD_{V\oplus U}=\cD_{V_U}\cD_{U}$, completing the proof.
\end{proof}


\begin{lem}\label{lem:super-self-dual}
If $W$ is the superkernel of a symmetric $(M|N)$-rational pseudo-differential operator $\cR$\textbf{}
then $V^*=V_U$ and $U^*=U_V$. In particular, $W$ is a normalized self-dual space, $V=U^\perp$, $U=V^\perp$, $\Wr(V)=\Wr(U)$, and $\Wr(W)=1$.
\end{lem}
\begin{proof}
Since $\cR$ is symmetric, we have
\[
\cD_V^{}\cD_U^{-1}=\cR=\cR^*=(\cD_{U_V}^{-1}\cD_{V_U}^{})^*=\cD_{V_U}^*(\cD_{U_V}^{-1})^*=\cD_{V_U}^*\big(\cD_{U_V}^*\big)^{-1}.
\]Because the minimal fractional factorization is unique, we conclude that $\cD_V^{}=\cD_{V_U}^*$ and $\cD_U=\cD_{U_V}^*$. We have $V=(V_U)^*$ and $U=(U_V)^*$. Note that $*$ is an involution, we have $V^*=V_U$ and $U^*=U_V$.

Using \eqref{eq:V-U} and Lemma \ref{lem:dual}, it follows from \cite[Lemma A.4]{MV04} that
\beq\label{eq:V-Ubasis}
(V_U)^*=\left\langle\frac{\Wr(u_1,\dots,u_N,v_1,\dots,\widehat{v_i},\dots,v_M)}{\Wr(u_1,\dots,u_N,v_1,\dots,v_M)}\right\rangle_{1\lle i\lle M},
\eeq
\beq\label{eq:U-Vbasis}
(U_V)^*=\left\langle\frac{\Wr(v_1,\dots,v_M,u_1,\dots,\widehat{u_j},\dots,u_N)}{\Wr(v_1,\dots,v_M,u_1,\dots,u_N)}\right\rangle_{1\lle j\lle N}.
\eeq
Since $V=(V_U)^*$ and $U=(U_V)^*$, it follows from Lemma \ref{lem:dual} that $W$ is a normalized self-dual space. Therefore by the same lemma,  $\Wr(W)=1$. It is also clear from \eqref{eq:V-Ubasis} and \eqref{eq:U-Vbasis} that $V=U^\perp$ and $U=V^\perp$.

Computing the Wronskian of $(V_U)^*$ using basis \eqref{eq:V-Ubasis}, by \cite[Lemma A.4]{MV04} and $\Wr(W)=1$, we have $\Wr((V_U)^*)=\Wr(U)$. Therefore, $\Wr(V)=\Wr(U)$.
\end{proof}
\begin{rem}
The statement $\Wr(V)=\Wr(U)$ also follows from Lemma \ref{lem:wr-orho-complement}.\qed
\end{rem}

\begin{lem}\label{no odd}
If $R$ is a symmetric $(M|N)$-rational pseudo-differential operator, then either $M$ or $N$ is even.
\end{lem}
\begin{proof}
If $M+N$ is even, then the canonical form is nondegenerate and skew-symmetric. If in addition $M, N$ are odd, then the restriction of the canonical form to $V$ and $U$ must be nondegenerate and skew-symmetric which is impossible.
\end{proof}
 
Let $\cR$ be a symmetric $(M|N)$-rational pseudo-differential operator with the superkernel $W$. Without loss of generality let $M$ be even, see Lemma \ref{no odd}. Let $(\cdot,\cdot)$ be the canonical bilinear form on $W$. A superflag $\mathscr F=\{F_1\subset \dots\subset F_{M+N}=W\}$ of $W$ is called \emph{isotropic} if $F_i=F_{M+N-i}^\perp$ for $1\lle i\lle M+N-1$. We denote the set of all isotropic superflags by $\sF^\perp(W)$. We note that $\dim(\sF^\perp(W))=M^2/4+N^2/4$ if $N$ is even and  $\dim(\sF^\perp(W))=M^2/4+(N-1)^2/4-M/2$ if $N$ is odd.

We call a complete $\bm s$-factorization $\cR=(\partial-f_1)^{s_1}\cdots (\pa-f_{M+N})^{s_{M+N}}$ of $\cR$ \textit{symmetric} if $f_{M+N+1-i}=-f_i$ and $s_{M+N+1-i}=s_i$ for $i\in I$. Denote $\cF^\perp(\cR)$ the set of all symmetric complete factorizations.

Recall the map $\varpi$ that gives a bijection between superflags $\sF(W)$ and complete factorizations $\cF(\cR)$, see Lemma \ref{lem:A-bij-flag-factor}. Denote $\varpi^\perp$ the restriction of this map to $\cF^\perp(\cR)\subset\cF(\cR)$.

\begin{prop}\label{prop:bij-f-iso}
The map $\varpi^\perp$ is a bijection between the set of all isotropic superflags $\sF^\perp(W)$ and the set of all symmetric complete factorizations $\cF^\perp(\cR)$.
\end{prop}
\begin{proof} 
If $\sF$ is an isotropic superflag, then $\varpi(\sF)$ is a symmetric complete factorization, since functions $f_i$ are given by \eqref{eq wronski coeff} and $f_i=-f_{M+N+1-i}$ follows by Lemma \ref{lem:wr-orho-complement}.

Given a symmetric complete factorization of $\mc R$, it corresponds to a superflag $\sF=\{F_1\subset \dots\subset F_{M+N}\}$ of $W$. We show that $F_i^\perp=F_{M+N+1-i}$ by induction on $i$. The induction step follows from the following statement. If $W_1\subset W_2\subset W$ and $\dim W_2=\dim W_1+1$, then $W_2^\perp$ is the unique codimension one subspace of $W_1^\perp$ whose Wronskian is $\Wr(W_2)$.
\end{proof}

\begin{lem}\label{lem:swit-chpm}
We have
\beq\label{eq:switchpm}
(\pa-f)\pa^{-1}(\pa+f)=(\pa+f)\pa^{-1}(\pa-f),\qquad f\in \cK.
\eeq
\end{lem}
\begin{proof}
The statement follows from direct computation.
\end{proof}

This lemma will be important for us to understand the Bethe ansatz equation for $\osp_{2m|2n}$, see Lemma \ref{lem fake-inv} below.

\section{Bethe ansatz}\label{sec: bethe}

\subsection{Bethe ansatz equation}
Let $\tlbms\in \wt S_{m|n}^\iota$ be an extended parity sequence, $\bla=(\la_1,\dots,\la_p)$ a sequence of dominant integral $\osp_{2m+\iota|2n}$-weights, $\bm z=(z_1,\dots,z_p)$ a sequence of pairwise distinct complex numbers. Let $\bm l=(l_1,\dots,l_{r})$ be a sequence of non-negative integers. Set $l=\sum_{i=1}^r l_i$. 

Define the \emph{color function} $c: \{1,\dots,l\}\to I$ associated with $\bm l$ by
\[
c(j)=a,\qquad \text{ if } \sum_{i=1}^{a-1}l_i < j\lle \sum_{i=1}^a l_i.
\]
Let $\bm t=(t_1,\dots,t_l)$ be a set of variables. We say that the variable $t_i$ has the \emph{color} $c(i)$.

The \emph{Bethe ansatz equation} associated with $\bla$, $\bm z$, $\bm l$ and $\tlbms$ is a system of algebraic equations in variables $\bm t$:
\beq\label{eq:BAE}
-\sum_{j=1}^p \frac{(\la^{\tlbms}_j,\alpha_{c(i)}^{\tlbms})}{t_i-z_j}+\sum_{a=1,a\ne i}^l\frac{(\alpha_{c(a)}^{\tlbms},\alpha_{c(i)}^{\tlbms})}{t_i-t_a}=0,\qquad 1\lle i\lle l.
\eeq

The Bethe ansatz equations for the case of $\glMN$ are the same, except that one has to drop tildes and change $r$ to $r-1$, see \cite{MVY15}.

Suppose $\bm t$ is a solution of the Bethe ansatz equation. If $(\alpha_{c(a)}^{\tlbms},\alpha_{c(i)}^{\tlbms})\ne 0$ and $i\ne a$, then we have $t_i\ne t_a$. If   $(\la_j^{\tlbms},\alpha_{c(i)}^{\tlbms})\ne 0$, then $t_i\ne z_j$. 

Following \cite{HMVY19}, we impose the following condition, see Theorem \ref{thm:wr-div} below. Suppose $(\alpha_j^{\tlbms},\alpha_j^{\tlbms})=0$ for some $j$. Fix all other $t_r$ whose colors are different from $j$. Consider the single equation \eqref{eq:BAE} for some $t_i$ with color $j$, then this equation does not depend on $j$ but on $c(j)=i$. Let $t_0$ be a solution to this equation with multiplicity $b$. Then we require that the number of $t_a$ for $1\lle a\lle l$ such that $t_a=t_0$ and $c(a)=i$ is at most $b$.

Let $\fkS_a$ be the symmetric group permuting $\{1,\dots,a\}$. Clearly, the group $\fkS_{\bm l}:=\fkS_{l_1}\times \dots\times \fkS_{l_r}$ acts on $\bm t$ by permuting the variables of the same color. We do not distinguish between solutions of the Bethe ansatz equation in the same $\fkS_{\bm l}$-orbit.

\medskip

Equation \eqref{eq:BAE} corresponds to the so called periodic boundary conditions. One can consider quasi-periodic boundary conditions by adding generic constants $u_i$ to the right hand side. Most of the result in this paper can be reworked for this case, cf. \cite{MV08}.  We do not discuss quasi-periodic boundary conditions any further.

\subsection{Polynomials representing solution of Bethe ansatz equation}
Let $\bla=(\la_1,\dots,\la_p)$ be a sequence of dominant integral $\osp_{2m+\iota|2n}$-weights, $\bm z=(z_1,\dots,z_p)$ a sequence of pairwise distinct complex numbers. Fix an extended parity sequence $\tlbms\in \wt S^\iota_{m|n}$.

Define a sequence of rational functions $\bm{P}^{\tlbms}=(P_1^{\tlbms},\dots,P_{r}^{\tlbms})$ \emph{associated with the data $\bla$, $\bm z$, and $\tlbms$},
\beq\label{eq:p-polynomials}
P_{i}^{\tlbms}(x)=\prod_{a=1}^p (x-z_a)^{s_i(\la_a^{\tlbms},\ve_i^{\tlbms})},\qquad i\in I.
\eeq

Define a sequence of rational functions $\bm T^{\tlbms}=(T_1^{\tlbms},\dots,T_{r}^{\tlbms})$ \emph{associated with the data $\bla$, $\bm z$, and $\tlbms$},
\beq\label{eq:T-polynomials}
\begin{split}
&T_{i}^{\tlbms}(x)= P_i^{\tlbms}(P_{i+1}^{\tlbms})^{-s_is_{i+1}},\qquad 1\lle i\lle r-1,\\
&T_{r}^{\tlbms}(x)= \begin{cases}
(P_r^{\tlbms})^{2}, & \text{ if }\iota=1,\\
P_{r-1}^{\tlbms}(P_{r}^{\tlbms})^{s_{r-1}}, & \text{ if }\iota=0 \text{ and }\bm s \text{ is of type D},\\
P_{r}^{\tlbms}, & \text{ if }\iota=0 \text{ and }\bm s \text{ is of type C}.
\end{cases}
\end{split}
\eeq
Note that $T_{i}^{\tlbms}$ are polynomials.

Note that conversely, $P_i^{\tlbms}$ can be written in terms of $T_j^{\tlbms}$.

The rational functions $P_i^{\tlbms}$ completely determine the nonzero weights $\la_j^{\tlbms}$, $1\lle j\lle p$. We often assume $P_i^{\tlbms}$ are given meaning $\bs\la$ is given. 

In the case of $\glMN$, the polynomials $P_i^{\bm s}$, $i\in I$, are defined by the same equation \eqref{eq:p-polynomials} dropping the tildes. The polynomials $T_i^{\bm s}$, $1\lle i\lle r-1$, are also given by \eqref{eq:T-polynomials} dropping the tildes and ignoring $T_r^{\bm s}$. In this case, $P_i^{\bm s}$ can be written in terms of $T_j^{\bm s}$ and $P_r^{\bm s}$. 
 
Let $\bm l=(l_1,\dots,l_{r})$ be a sequence for non-negative integers and $\bm t=(t_i)_{i=1}^l$ a sequence of complex numbers. Define a sequence of polynomials $\bm y=(y_i)_{i=1}^r$ by
\beq\label{eq:y-polynomials}
y_i(x)=\ \prod_{\mathclap{j=1,c(j)=i}}^l\ (x-t_j), \qquad i\in I.
\eeq
Then we say that the \emph{sequence of polynomials $\bm y$ represents the sequence of complex numbers} $\bm t$. By definition of color function, we have $\deg y_i=l_i$.

Using $\bm{T}^{\tlbms}$ and $\bm y$, the Bethe ansatz equation \eqref{eq:BAE} can be written as the vanishing conditions,
\beq\label{eq:BAE-poly-T}
\begin{split}
&-\frac{(T_i^{\tlbms})'}{T_{i}^{\tlbms}}+\sum_{j=1,j\neq i}^r c_{ij}^{\bm s}\frac{y_j'}{y_j}=0,\quad\qquad\ \, \text{ if }c_{ii}^{\bm s}=0,\\
&-\frac{(T_i^{\tlbms})'}{T_{i}^{\tlbms}}+\frac{y_i''}{y_i'}+\sum_{j=1,j\neq i}^r c_{ij}^{\bm s}\frac{y_j'}{y_j}=0,\quad \text{ if }c_{ii}^{\bm s}\ne 0,
\end{split}
\eeq
when evaluated at roots of $y_i(x)$.

\section{Reproduction procedure}\label{sec: reproduction}
Given a solution of Bethe ansatz equation \eqref{eq:BAE}, and an $\tlbms$-simple root, we construct a family of new solutions. Such a construction is called a reproduction procedure. The reproduction procedure for roots of nonzero length was defined in \cite{MV04}. It does not change the Cartan matrix. For roots of zero length, the reproduction procedure was defined in \cite{HMVY19} for the special case of $\gl_{m|n}$. We describe the most general form of fermionic reproduction procedure and show how it changes the Cartan matrix and the weights. Then we apply this result to the case of $\osp_{2m+\iota|2n}$.

\subsection{General Bethe ansatz equation}
We call an $r\times r$ matrix $C=(c_{ij})_{i,j\in I}$ a \emph{generalized Cartan matrix} if the following conditions are satisfied:
\begin{itemize}
\item $c_{ij}=0$ if and only if $c_{ji}=0$;
\item $c_{ii}=0$ or $c_{ii}=2$; 
\item there exist non-zero integers $d_i$, $i\in I$, such that $d_ic_{ij}=d_jc_{ji}$.
\end{itemize}
Note that we do not require all $c_{ij}$ to be integers.

Let $C=(c_{ij})_{i,j\in I}$ be a generalized Cartan matrix. 
Let $\bm T=(T_1,\dots,T_r)$ be a sequence of functions of the form
$T_i=\prod_{a=1}^k(x-z_a)^{\mu_a}$, where $z_a,\mu_a\in\C$.

The Bethe ansatz equation associated with $(C,\bm T)$ on the set of zeroes of polynomials $\bs y$ is the condition of vanishing
\beq\label{eq:BAE-poly-T-nos}
\begin{split}
&-\frac{T_i'}{T_{i}}+\sum_{j=1,j\neq i}^r c_{ij}\frac{y_j'}{y_j}=0,\quad\qquad\ \, \text{ if }c_{ii}=0,\\
&-\frac{T_i'}{T_{i}}+\frac{y_i''}{y_i'}+\sum_{j=1,j\neq i}^r c_{ij}\frac{y_j'}{y_j}=0,\quad \text{ if }c_{ii}\ne 0,
\end{split}
\eeq
on the roots of $y_i$, cf. \eqref{eq:BAE-poly-T}.

\begin{rem}\label{DC remark}
Let $D$ be an $r\times r$ diagonal matrix with rational non-zero diagonal entries $d_i$, $i\in I$, such that $d_i=1$ if $c_{ii}=2$. Then the Bethe ansatz equations associated with $(C,\bm T)$ and $(DC,\bm T^D)$, where $\bm T^D=(T_1^{d_1},\dots,T_r^{d_r})$, coincide.
\end{rem}

\medskip

We say that the pair $(C,\bm T)$ is \emph{admissible} if for all $i$ such that $c_{ii}=2$, then $c_{ij}\in\Z_{\lle 0}$ for all $j\ne i$ and the function $T_i$ is a polynomial.
In this case, following \cite{MV04} and \cite{HMVY19}, we rewrite the Bethe ansatz equations as follows.

Given $f(x)=\prod_{a=1}^k(x-z_a)^{\mu_a}$, with distinct $z_a$, define a polynomial
$$
\pi(f)=\prod_{a:\, \mu_a\neq 0}(x-z_a).
$$
The polynomial $\pi(f)$ is the  monic denominator of the rational function $\ln'(f)$ of minimal possible degree.

Set $\pi_i=\pi(T_i)$.

We say that a sequence of polynomials $\bm y=(y_1,\dots,y_r)$ is \emph{generic with respect to $(C,\bm T)$}, if the following conditions are satisfied:
\begin{enumerate}
    \item if $c_{ii}\ne 0$, then $y_i$ has no multiple roots;
    \item if $c_{ij}\ne 0$, then $y_i$ and $y_j$ have no common roots;
    \item the zeros of $y_i$ are not zeroes of $\pi_i$.
\end{enumerate}
In particular, if some polynomial $y_i$ is a zero polynomial, then the sequence $\bm y$ is not generic.

If the sequence of polynomials $\bm y$ represents a solution of Bethe ansatz equation associated with $(C,\bs T)$, then $\bm y$ is generic with respect to $(C,\bs T)$.

The following theorem restates the Bethe ansatz equation in terms of Wronskian equalities and divisibility conditions.

\begin{thm}\label{thm:wr-div}
Let $\bm y=(y_1,\dots,y_{r})$ be a sequence of polynomials generic with respect to an admissible pair $(C,\bs T)$. Then the sequence $\bm y$ represents a solution of the Bethe ansatz equation \eqref{eq:BAE-poly-T-nos} if and only if for each $1\lle i\lle r$ there exists a polynomial $\tl y_i$  satisfying
\beq\label{eq:wr-div}
\begin{split}
   &\Wr(y_i,\tl y_i)=T_i \prod_{j=1,j\ne i}^{r}y_j^{-c_{ij}} ,\qquad \qquad\qquad\ \ \, \quad\text{ if }c_{ii} =2,\\
    &y_i\tl y_i=\ln'\Big(T_i \prod_{j=1}^{r}y_j^{-c_{ij}} \Big)\pi_i \prod_{j=1:\,c_{ij} \ne 0}^{r}y_j,\qquad\quad\ \,  \text{ if }c_{ii} =0.
\end{split}
\eeq
\end{thm}
\begin{proof}
The statement follows from \cite[Lemma 3.2]{MV04} and \cite[Lemma 5.3]{HMVY19}.
\end{proof}

Let $\bm y=(y_1,\dots, y_{r})$ be a sequence polynomials. We say that $\bm y$ is \emph{fertile in the $i$-th direction with respect to} $(C,\bm T)$ if there exists a nonzero polynomial $\tl y_i$ satisfying \eqref{eq:wr-div}. We call $\bm y$ {\it fertile with respect to} $(C,\bm T)$ if $\bm y$ is fertile in all directions with respect to $(C,\bm T)$. Note that $\bs y$ being fertile does not imply that $\bm y$ is generic.

\subsection{Reproduction procedure}\label{sec:general-rep}
Let $\bm y=(y_1,\dots,y_{r})$ represent a solution of the Bethe ansatz equation \eqref{eq:BAE-poly-T-nos} associated with admissible $(C,\bs T)$. Choose $i\in I$ and let $\tilde y_i$ (may not be unique) be as in Theorem \ref{thm:wr-div}. Set 
$$\bm y^{[i]}=(y_1,\dots,\tl y_i,\dots, y_r).$$ 
We call $\bm y^{[i]}$ an \emph{immediate descendant} of $\bm y$ in the $i$-th direction. We show that under some generic conditions, an immediate descendant also represents a solution of the Bethe ansatz equation associated with possibly new data $( C^{[i]}, \bm T^{[i]})$. 


Fix $i\in I$. Define $ C^{[i]}= ( c^{[i]}_{j_1j_2})_{j_1,j_2\in I}$ and $\bm T^{[i]}=(  T_1^{[i]},\dots, T_r^{[i]})$ as follows. In the case of $c_{ii}\neq 0$, then there is no change, $( C^{[i]},\bs T^{[i]})=(C,T)$. If $c_{ii}=0$, then the data changes according to the following rules.
\begin{itemize}
    \item if $j$ is such that $c_{ij}=0$, then the $j$-th row of $C$ does not change $ c_{jk}^{[i]}=c_{jk}$, and, moreover, $ T_j^{[i]}=T_j$;
    \item if $j$ is such that $c_{ij}\ne 0$, then $c_{jk}^{[i]}$ and $ T_j^{[i]}$ are given in the Table \ref{table:1}.
\end{itemize}

\begin{table}[h!] \centering
 \begin{tabular}{|c|c|c|}
 \hline
 $C$  & $ C^{[i]}$ & $ T_{j}^{[i]}$\\
 \hline
     $\begin{pmatrix}
 q_1 & 0 & q_2\\
 -q_3 & -1 & 2
 \end{pmatrix}$  & $\begin{pmatrix}
 q_1 & 0 & q_2\\
 q_1+q_2+q_2q_3 & -q_2 & 0
 \end{pmatrix}$ & $ T_{j}^{[i]}=\dfrac{T_i}{T_{j}^{q_2}\pi_{i}^{q_2}}$\\
  \hline
 $\begin{pmatrix}
 q_1 & 0 & q_2\\
 -q_3 & -q_4 & 0
 \end{pmatrix}$  & $\begin{pmatrix}
 q_1 & 0 & q_2\\
 \frac{q_3}{q_4}+\frac{q_1}{q_2}+1 & -1 & 2
 \end{pmatrix}$& $ T_{j}^{[i]}=\dfrac{T_{i}^{1/q_2}}{\pi_{i}T_{j}^{1/q_4}}$\\
 \hline
  $\begin{pmatrix}
0 & 0 & q_2\\
 -q_3 & -1 & 2
 \end{pmatrix}$ & $\begin{pmatrix}
 0 & 0 & q_2\\
q_2q_3 & -q_2 & 0
 \end{pmatrix}$& $ T_{j}^{[i]}=\dfrac{T_i}{T_{j}^{q_2}\pi_{i}^{q_2}}$\\
    \hline
 $\begin{pmatrix}
 0 & 0 & q_2\\
 -q_3 & -q_4 & 0
 \end{pmatrix}$ & $\begin{pmatrix}
 0 & 0 & q_2\\
 \frac{q_3}{q_4}  & -1 & 2
 \end{pmatrix}$& $ T_{j}^{[i]}=\dfrac{T_{i}^{1/q_2}}{\pi_{i}T_{j}^{1/q_4}}$\\
  \hline
 $\begin{pmatrix}
 q_1 & 0 & q_2\\
 -q_3 & -q_4-1 & 2
 \end{pmatrix}$  & $\begin{pmatrix}
 q_1 & 0 & q_2\\
 \frac{q_1(q_4+1)+q_2(q_3+q_4+1)}{q_2q_4} & -\frac{q_4+1}{q_4} & 2
 \end{pmatrix}$& $ T_{j}^{[i]}=\left(\dfrac{T_{i}^{(q_4+1)/q_2}}{\pi_{i}^{q_4+1}T_{j}}\right)^{1/q_4}$\\
  \hline
 $\begin{pmatrix}
 0 & 0 & q_2\\
 -q_3 & -q_4-1 & 2
 \end{pmatrix}$ & $\begin{pmatrix}
 0 & 0 & q_2\\
 \frac{q_3}{q_4} & -\frac{q_4+1}{q_4} & 2
 \end{pmatrix}$& $ T_{j}^{[i]}=\left(\dfrac{T_{i}^{(q_4+1)/q_2}}{\pi_{i}^{q_4+1}T_{j}}\right)^{1/q_4}$\\
 \hline
 \end{tabular}
 \bigskip
 \caption{The change of Cartan matrix and weights.} \label{table:1}
\end{table}
In Table \ref{table:1}, the shown $2\times 3$ submatrices are the submatrices of $C$ of the form.
$$\begin{pmatrix}
 c_{ik} & c_{ii} & c_{ij}\\
 c_{jk} & c_{ji} & c_{jj}
 \end{pmatrix}.
 $$
It is assumed that $q_1,q_2,q_4$ are non-zero integers.

We have $(C^{[i]})^{[i]}=C$ and $\bm T=(\bm T^{[i]})^{[i]}$.

\begin{lem}
If $C$ is a generalized Cartan matrix, then $C^{[i]}$ is a generalized Cartan matrix.
\end{lem}
\begin{proof}
The lemma is proved using case by case checking.
\end{proof}

Now we are ready to formulate the main result of this section.

Let $\bm y=(y_1,\dots,y_{r})$ be a sequence of polynomials representing a solution of the Bethe ansatz equation \eqref{eq:BAE-poly-T-nos} associated with admissible $(C,\bm T)$. Let $\bm y^{[i]}=(y_1,\dots,\tl y_i,\dots, y_r)$ be an immediate descendant of $\bm y$ in the $i$-th direction. Let $C^{[i]}$ and $\bm T^{[i]}$ be as given in Table \ref{table:1}. 
\begin{thm}\label{thm:general}
If $\bm y^{[i]}$ is generic with respect to $( C^{[i]},\bm T^{[i]})$, then $\bm y^{[i]}$ satisfies the Bethe ansatz equation \eqref{eq:BAE-poly-T-nos} associated with $(C^{[i]},\bm T^{[i]})$. 
\end{thm}
\begin{proof}
We write the proof for the case corresponding to the first row of Table \ref{table:1}. The other cases are similar. 

We say a rational function $f(x)$ is zero modulo polynomial $y(x)$ if $f(x)=p(x)/q(x)$ where $p(x)$ and $q(x)$ are relatively prime and $y(x)$ divides $p(x)$.  

From \eqref{eq:BAE-poly-T-nos} we have modulo $y_{j}$
\beq\label{eq:yy2}
\frac{y_{j}''}{y_{j}'}\equiv\frac{T_{j}'}{T_{j}}+q_3\frac{y_{k}'}{y_{k}}+\frac{y_i'}{y_i}+\cdots.
\eeq
Here and below by the dots we mean the terms which in the final 
result \eqref{the last one} create terms which depend only on $y_a$ with $a\neq i,j,k$.

Since $(C,\bm T)$ is admissible, by Theorem \ref{thm:wr-div}, there exists a polynomial $\tl y_i$ such that
\beq\label{eq:pf-rep-1} 
y_i\tl y_i=\ln'\Big(\frac{T_{i}}{y_{k}^{q_1}y_{j}^{q_2}}\Big)\pi_iy_{k}y_{j}+\cdots=\Big(\frac{T_{i}'}{T_{i}}-q_1\frac{y_{k}'}{y_{k}}\Big)\pi_i  y_{k}y_{j}-q_2\pi_i y_{k}y_{j}'+\cdots.
\eeq

Differentiating both sides, we have modulo $y_j$
\begin{align*}
y_i'\tl y_i+y_i\tl y_i'\equiv\Big(\frac{T_{i}'}{T_{i}}-q_1\frac{y_{k}'}{y_{k}}\Big)\pi_i  y_{k}y_{j}'-q_2\pi_i y_{k}y_{j}''-q_2\pi_i' y_{k}y_{j}'-q_2\pi_i y_{k}'y_{j}'+\cdots.
\end{align*}
Here and below we use that $\bs y$ is generic with respect $(C,\bm T)$.
Dividing further by \eqref{eq:pf-rep-1}, we have modulo  $y_{j}$
\[
\frac{y_i'}{y_i}+\frac{\tl y_i'}{\tl y_i}\equiv -\frac{1}{q_2}\Big(\frac{T_{i}'}{T_{i}}-q_1\frac{y_{k}'}{y_{k}}\Big)+\frac{y_{j}''}{y_{j}'}+\frac{y_{k}'}{y_{k}}+\frac{\pi_i'}{\pi_i}+\cdots.
\]
Adding \eqref{eq:yy2}, we obtain
\begin{equation}\label{the last one}
\Big(q_2\frac{T_{j}'}{T_{j}}+q_2\frac{\pi_i'}{\pi}-\frac{T_{i}'}{T_{i}}\Big)+(q_1+q_2+q_2q_3)\frac{y_{k}'}{y_{k}} - q_2\frac{\tl y_i'}{\tl y_i}+\cdots\equiv 0
\end{equation}
modulo $y_{j}$, completing the proof.
\end{proof}

If $c_{ii}\ne 0$, we can construct a family of sequences of polynomials $\bm y_c^{[i]}=(y_1,\dots,\tl y_i+cy_i,\dots, y_r)$. We call this construction the \emph{bosonic reproduction procedure in the $i$-th direction}.

If $c_{ii}= 0$, we can construct a single new sequence of polynomials $\bm y^{[i]}=(y_1,\dots,\tl y_i,\dots, y_r)$. We call this construction the \emph{fermionic reproduction procedure in the $i$-th direction}. Note that $\bm y^{[i]}$ is fertile in the $i$-th direction with respect to $(C^{[i]},\bm T^{[i]})$, and $(\bm y^{[i]})^{[i]}=\bm y$.

\subsection{Reproduction procedure for $\glMN$}

In this section, we apply Theorem \ref{thm:general} to the case of $\g=\gl_{m|n}$. Here we only need the first four rows of Table \ref{table:1} with $q_1=-1$, $q_2=q_4=1$, and $q_3=0$. Let $\bm s\in S_{m|n}$. 
Let $C^{\bm s}$ be the Cartan matrix of $\glMN$ associated with the parity sequence $\bm s$, then direct computations show that $(C^{\bm s})^{[i]}$ is exactly the Cartan matrix $C^{\bm s^{[i]}}$ of $\glMN$ associated with the parity sequence $\bm s^{[i]}$.

Let $\bla=(\la_1,\dots,\la_p)$ be a sequence of dominant integral $\glMN$-weights, $\bm z=(z_1,\dots,z_p)$ a sequence of pairwise distinct complex numbers. Let $\bm P^{\bm s}=(P_1^{\bm s},\dots,P_{r}^{\bm s})$ and $\bm T^{\bm s}=(T_1^{\bm s},\dots,T_{r-1}^{\bm s})$ be the sequences of rational functions associated with $\bla$, $\bm z$, and $\bm s$, see \eqref{eq:p-polynomials} and \eqref{eq:T-polynomials}. Note that in this case, $(C^{\bs s},\bm T^{\bs s})$ is admissible for all $\bm s\in S_{m|n}$.

We have $T_i^{\bm s}=P_i^{\bm s}(P_{i+1}^{\bm s})^{-s_is_{i+1}}$. 

Denote $\pi(T_i^{\bm s})$ by $\pi_i^{\bm s}$. 

For $1\lle i\lle  r-1$, recall that $\bm s^{[i]}=(s_1,\dots,s_{i+1},s_i,\dots,s_r)$. Clearly, if $s_i=s_{i+1}$, then $\bm P^{\bm s^{[i]}}=\bm P^{\bm s}$ and if $s_i\ne s_{i+1}$, then
\beq\label{eq:P-change}
\bm P^{\bm s^{[i]}}=(P_1^{\bm s},\dots,P_{i+1}^{\bm s}\pi_i^{\bm s},P_i^{\bm s}(\pi_i^{\bm s})^{-1},\dots,P_{r}^{\bm s}),
\eeq
see \cite[Lemma 6.1]{HMVY19} (note that $P_i^{\bm s}$ are denoted by $T_i^{\bm s}$ there) and cf. Lemma \ref{lem: weight change}.  

It is straightforward to check that $(\bm T^{\bm s})^{[i]}$ obtained from Table \ref{table:1} coincides with the sequence $\bm T^{\bm s^{[i]}}$ of rational functions associated with $\bla$, $\bm z$, and $\bm s^{[i]}$.

Then in the case of $\gl_{m|n}$, Theorem \ref{thm:general} and the reproduction procedure take the following form. Set $y_0(x)=y_{r}(x)=1$.
\begin{thm}[{\cite[Theorem 6.2]{HMVY19}}]\label{thm:A-rep}
Let $\bm y=(y_1,\dots,y_{r-1})$ be a sequence of polynomials generic with respect to $(C^{\bm s}, \bm T^{\bm s})$.
\begin{enumerate}
    \item The sequence $\bm y$ represents a solution of the Bethe ansatz equation \eqref{eq:BAE-poly-T-nos} associated with $(C^{\bm s}, \bm T^{\bm s})$, if and only if for each $1\lle i\lle r-1$, there exists a polynomial $\tl y_i$ satisfying
    \begin{align}
    &\Wr(y_i,\tl y_i)=P_i^{\bm s}(P_{i+1}^{\bm s})^{-1}y_{i-1}y_{i+1},\quad &\text{ if }s_i=s_{i+1},\label{eq:A-bos}\\
    &y_i\tl y_i=\ln'\Big(\frac{P_i^{\bm s}P_{i+1}^{\bm s}y_{i-1}}{y_{i+1}}\Big)\pi_i^{\bm s}y_{i-1}y_{i+1},\quad &\text{ if }s_i\ne s_{i+1}.\label{eq:A-fer}
    \end{align}
    \item  If $\bm y^{[i]}=(y_1,\dots,\tl y_i,\dots,y_{r-1})$ is generic with respect to $(C^{\bm s^{[i]}}, \bm T^{\bm s^{[i]}})$, then $\bm y^{[i]}$ represents a solution of the Bethe ansatz equation associated with $(C^{\bm s^{[i]}}, \bm T^{\bm s^{[i]}})$.\qed
\end{enumerate}
\end{thm}
Note that in the theorem we can remove all $P_j^{\bm s}$ and replace them with a single $T_i^{\bm s}$. For example, if $s_i=s_{i+1}$, equation \eqref{eq:A-bos} is equivalent to
$$
\Wr(y_i,\tl y_i)=T_i^{\bm s} y_{i-1}y_{i+1}. $$
It is convenient to have the form with $P_i^{\bm s}$ for rational pseudo-differential operators we use below.

Note that the reproduction procedure in the $i$-th direction is bosonic if $s_i=s_{i+1}$, and fermionic if $s_i\neq s_{i+1}$. Bosonic reproduction procedure does not change parity (and therefore the Cartan matrix) and fermionic does.

If $\bm y^{[i]}$ is fertile in some direction with respect to $(C^{\bm s^{[i]}}, \bm T^{\bm s^{[i]}})$, then we can apply reproduction procedure again.

Denote by $\mathscr P_{\bm y,\bm s}^{A}\subset (\mathbb P(\C[x]))^{r-1}\times S_{m|n}$
 the set of all pairs obtained from the initial pair $(\bm y,\bm s)$ by repeatedly applying  reproduction procedure as much as possible. We call $\mathscr P^{A}_{\bm y,\bm s}$ the $\gl_{m|n}$ \emph{population of solutions of the Bethe ansatz equation} originated from $(\bm y,\bm s)$.

There is a rational pseudo-differential operator which is invariant under reproduction procedure. Let $\bm y$ be as in Theorem \ref{thm:A-rep}. Define a rational pseudo-differential operator $\mathscr R_{\bm y,\bm s}$ over $\C(x)$
\beq\label{eq:diff-oper}
\mathscr R_{\bm y,\bm s}=\mathop{\overrightarrow\prod}\limits_{1\lle i\lle r}\Big(\pa-s_i\ln'\frac{P_i^{\bm s}y_{i-1}}{y_i}\Big)^{s_i}.
\eeq
\begin{thm}[{\cite[Theorem 6.3]{HMVY19}}]\label{thm:A-inv}
Let $\mathscr P$ be a $\gl_{m|n}$ population. Then the rational pseudo-differential operator $\mathscr R_{\bm y,\bm s}$ is independent of the choice of the pair $(\bm y,\bm s)$ in $\mathscr P$.\qed
\end{thm}

\medskip

Let $\bm y$ represent a solution of the Bethe ansatz equation associated with $(C^{\bm s_+},T^{\bm s_+})$. 
Suppose at least one of the weights $\la_i$ is typical. Then $\mathscr R_{\bm y,\bm s_+}$ is an $(m|n)$-rational pseudo-differential operator, cf. \cite[Proposition 7.7]{HMVY19}. Let $W=V\oplus U$ be the superkernel of $\mathscr R_{\bm y,\bm s_+}$. We have $\dim V=m$, $\dim U=n$, $\dim W=m+n$. 

\begin{thm}[{\cite[Theorem 7.9]{HMVY19}}]\label{thm:A-bijection}
Assume $\mathscr R_{\bm y,\bm s_+}$ is an $(m|n)$-rational pseudo-differential operator. Then there exist bijections between the $\glMN$ population $\mathscr P_{\bm y,\bm s_+}^A$,  the set of complete factorizations  $\mc F(\mathscr R_{\bm y,\bm s_+})$, and the space $\sF(W)$ of full superflags in $W$.\qed
\end{thm}

\begin{rem}
In \cite{HMVY19}, this  theorem is formulated for the case of polynomial weights only. It also applies if $\bs y$ is assumed to be superfertile, see Section \ref{sec:superfertility}. Moreover, the population $\mathscr P_{\bm y,\bm s_+}^A$ is always embedded in the set of  complete factorizations  $\mc F(\mathscr R_{\bm y,\bm s_+})$ without any additional assumptions.
\end{rem}

\subsection{Reproduction procedure for $\osp_{2m+1|2n}$}\label{sec b}
In this section, we apply Theorem \ref{thm:general} to the case of $\g=\osp_{2m+1|2n}$. Since this section deals only with the case $\iota=1$, we work with parity sequences and not extended parity sequences. 

Let $\bm s\in S_{m|n}$.
Recall sequences $\bm s^{[i]}=(s_1,\dots,s_{i+1},s_i,\dots,s_r)$, $1\lle i\lle r-1,$
 and  set $\bm s^{[r]}=\bm s$.

Let $C^{\bm s}$ be the Cartan matrix of $\g$ associated with the parity sequence $\bm s$, then, as in the case of $\glMN$, direct computations show that $(C^{\bm s})^{[i]}$ is exactly the Cartan matrix $C^{\bm s^{[i]}}$ of $\g$ associated with the parity sequence $\bm s^{[i]}$. 

The only new case comparing to the $\glMN$ case is that the fermionic reproduction procedure in the $(r-1)$-st direction.
Here we use the last two rows of Table \ref{table:1} with $q_1=-1$, $q_2=q_4=1$, and $q_3=0$. In particular, the fermionic reproduction procedure in the $(r-1)$-st direction
does not change the $r$-th row of the Cartan matrix. Note that this implies the reproduction procedure in the $r$-th direction is always bosonic.

Let $\bla=(\la_1,\dots,\la_p)$ be a sequence of dominant integral $\osp_{2m+1|2n}$-weights, $\bm z=(z_1,\dots,z_p)$ a sequence of pairwise distinct complex numbers. Let $\bm P^{\bm s}=(P_1^{\bm s},\dots,P_{r}^{\bm s})$ and $\bm T^{\bm s}=(T_1^{\bm s},\dots,T_{r}^{\bm s})$ be the sequences of rational functions associated with $\bla$, $\bm z$, and $\bm s$, see \eqref{eq:p-polynomials} and \eqref{eq:T-polynomials}. Note that in this case, $(C^{\bs s},\bm T^{\bs s})$ is admissible for all $\bm s\in S_{m|n}$.

We have $T_i^{\bm s}=P_i^{\bm s}(P_{i+1}^{\bm s})^{-s_is_{i+1}}$ for $1\lle i\lle r-1$ and $T_r^{\bm s}=(P_r^{\bm s})^2$.

Denote $\pi(T_i^{\bm s})$ by $\pi_i^{\bm s}$. 

Using Lemma \ref{lem: weight change}, we write $\bm P^{\bm s^{[i]}}$ in terms of $\bm P^{\bm s}$. For $1\lle i\lle  r-1$, if $s_i=s_{i+1}$, then $\bm P^{\bm s^{[i]}}=\bm P^{\bm s}$. Also, $\bm P^{\bm s^{[r]}}=\bm P^{\bm s}$.  If $s_i\ne s_{i+1}$, $1\lle i\lle  r-1$, then
\beq\label{eq:P-change-1}
\bm P^{\bm s^{[i]}}=(P_1^{\bm s},\dots,P_{i-1}^{\bm s},P_{i+1}^{\bm s}\pi_i^{\bm s},P_i^{\bm s}(\pi_i^{\bm s})^{-1},P_{i+2}^{\bm s},\dots,P_{r}^{\bm s}).
\eeq

Now it is straightforward to check that $(\bm T^{\bm s})^{[i]}$ obtained from Table \ref{table:1} coincides with the sequence $\bm T^{\bm s^{[i]}}$ of rational functions associated with $\bla$, $\bm z$, and $\bm s^{[i]}$.

Then in the case of $\osp_{2m+1|2n}$, Theorem \ref{thm:general} and the reproduction procedure take the following form. Set $y_0(x)=1$.

\begin{thm}\label{thm:B-rep}
Let $\bm y=(y_1,\dots,y_{r})$ be a sequence of polynomials generic with respect to $(C^{\bm s},\bm T^{\bm s})$.
\begin{enumerate}
    \item The sequence $\bm y$ represents a solution of the Bethe ansatz equation \eqref{eq:BAE-poly-T} associated with $(C^{\bm s},\bm T^{\bm s})$, if and only if there exist polynomials $\tl y_i$, $1\lle i\lle r$, satisfying \eqref{eq:A-bos}, \eqref{eq:A-fer} for $1\lle i\lle r-1$, and
    \begin{align}
    \Wr(y_r,\tl y_r)=(P_r^{\bm s})^2 y_{r-1}^2. \label{eq:B-bos2}
    \end{align}
    \item If $\bm y^{[i]}=(y_1,\dots,\tl y_i,\dots,y_{r})$ is generic with respect to $(C^{\bm s^{[i]}},\bm T^{\bm s^{[i]}})$, then $\bm y^{[i]}$ represents a solution of the Bethe ansatz equation \eqref{eq:BAE-poly-T} associated with $(C^{\bm s^{[i]}},\bm T^{\bm s^{[i]}})$.\qed
\end{enumerate}
\end{thm}

Similar to Theorem \ref{thm:A-rep}, in each case one can remove all $P_j^{\bm s}$ and write instead a single $T_i^{\bm s}$. 

For $1\lle i\lle r-1$, the reproduction procedure in the $i$-th direction is bosonic if $s_i=s_{i+1}$, and fermionic if $s_i\neq s_{i+1}$. The reproduction procedure in the $r$-th direction is always bosonic. As always, a bosonic reproduction does not change parity (and therefore the Cartan matrix) and fermionic does.

If $\bm y^{[i]}$ is fertile in some direction with respect to $(C^{\bm s^{[i]}}, \bm T^{\bm s^{[i]}})$, then we can apply reproduction procedure again.

Denote by $\mathscr P_{\bm y,\bm s}\subset (\mathbb P(\C[x]))^{r}\times S_{m|n}$
the set of all pairs $(\bar{\bm y},\bar {\bm s})$ obtained from the initial pair $(\bm y,\bm s)$ by repeatedly applying all possible reproduction procedures. We say that $\mathscr P_{\bm y,\bm s}$ is the $\osp_{2m+1|2n}$ \emph{population of solutions of the Bethe ansatz equation} originated from $(\bm y,\bm s)$.

\subsection{The differential operator of an $\osp_{2m+1|2n}$ population.}\label{sec B oper}

Our main idea to study
the $\osp_{2m+1|2n}$ populations is to include them in the  well-understood $\gl_{2m|2n}$ populations. The construction is as follows.

For a parity sequence $\bm s\in S_{m|n}$, define $\bm s^{A}\in S_{2m|2n}$ by
\[
\bm s^{A}=(s_1^A,\dots,s_{2r}^A)=(s_1,\dots,s_r,s_r,\dots,s_1),
\]
Let $C^{\bs s,A}$ be the Cartan matrix of $\gl_{2m|2n}$ associated to the parity sequence $\bm s^A$, cf. \eqref{eq:cartan-block}.

Let $\bm y=(y_1,\dots,y_{r})$ be a sequence of polynomials. Define a sequence of polynomials 
\beq\label{eq:y-A-B}
\bm y^{A}=(y_1^{A},\dots,y_{2r-1}^{A})=(y_1,y_2,\dots,y_{r-1},y_r,y_{r-1},\dots,y_1).
\eeq

Let $\bla$, $\bm P^{\bm s}$, and $\bm T^{\bm s}$ be as above. 

Define a sequence of dominant integral $\gl_{2m|2n}$-weights $\bla^A=(\la_1^{A},\dots,\la_p^{A})$ by
\[
\big((\la_i^A)^{\bm s^A},\ve^{\bm s^A}_j\big)=-\big((\la_i^A)^{\bm s^A},\ve^{\bm s^A}_{2r+1-j}\big)=(\la_i^{\bm s},\ve^{\bm s}_j),\qquad 1\lle i\lle p,\qquad 1\lle j\lle r.
\]

Define a sequence of rational functions $\bm P^{\bm s,A}$ by
\[
\bm P^{\bm s,A}=(P_1^{\bm s,A},\dots,P_{2r}^{\bm s,A})=(P_1^{\bm s},\dots, P_{r-1}^{\bm s},P_r^{\bm s},(P_r^{\bm s})^{-1},(P_{r-1}^{\bm s})^{-1},\dots,(P_1^{\bm s})^{-1})
\]
and a sequence of rational functions $\bm T^{\bm s,A}$ by
\[
\bm T^{\bm s,A}=(T_1^{\bm s,A},\dots,T_{2r-1}^{\bm s,A}),\quad T_i^{\bm s,A}=(T_{2r-i}^{\bm s,A})^{s_i^As_{i+1}^A}=T_i^{\bm s}.
\]
Note that $\bm P^{\bm s,A}$ and $\bm T^{\bm s,A}$ are exactly sequences of rational functions associated to $\bla^A$, $\bm z$, and $\bm s^A$, see \eqref{eq:p-polynomials} and \eqref{eq:T-polynomials}.
\begin{lem}\label{lem:AB}
The sequence of polynomials $\bm y$ represents a solution of the $\osp_{2m+1|2n}$ Bethe ansatz equation \eqref{eq:BAE-poly-T} associated with $(C^{\bs s},\bm T^{\bs s})$, if and only if  $\bm y^{A}$ represents a solution of the $\gl_{2m|2n}$ Bethe ansatz equation \eqref{eq:BAE-poly-T} associated with $(C^{\bm s,A},\bm T^{\bm s,A})$.
\end{lem}
\begin{proof}
The lemma is straightforward.
\end{proof}

Let $\bm y=(y_1,\dots,y_{r})$ be a sequence of polynomials. We have the  rational pseudo-differential operator $\mathscr R_{\bm y,\bm s}^{2m+1|2n}$ from \eqref{eq:diff-oper},
\beq\label{eq:B-oper}
\mathscr R_{\bm y,\bm s}^{2m+1|2n}=\mathscr R_{\bm y^A,\bm s^A}=  \mathop{\overrightarrow\prod}\limits_{1\lle i\lle 2r}\Big(\pa-s_i^A\ln'\frac{P_i^{\bm s,A}y^A_{i-1}}{y_i^A}\Big)^{s_i^A},
\eeq
where $y_0^A=y_{2r}^A=1$. Clearly, $\mathscr R_{\bm y,\bm s}^{2m+1|2n}$ is symmetric.

Suppose that $\bm y_0$ represents a solution of the $\osp_{2m+1|2n}$ Bethe ansatz equation associated with $(C^{\bs s_0},\bm T^{\bs s_0})$. Let $\mathscr P_{\bm y_{0},\bm s_0}$ be the $\osp_{2m+1|2n}$ population originated from $(\bm y_0,\bm s_0)$. Let $\mathscr P^A_{\bm y_0^A,\bm s_0^A}$ be the $\gl_{2m|2n}$ population originated from $(\bm y_0^A,\bm s_0^A)$, see Lemma \ref{lem:AB}.

\begin{prop}\label{prop:B-inv}
There exists an injective map $\mathscr P_{\bm y_0,\bm s_0}\to \mathscr P^A_{\bm y_0^A,\bm s_0^A}$, $(\bm y,\bm s)\mapsto (\bm y^A,\bm s^A)$. In particular, the rational pseudo-differential operator $\mathscr R^{2m+1|2n}_{\bm y,\bm s}$ is independent of the choice of $(\bm y,\bm s)$ in $\mathscr P_{\bm y_0,\bm s_0}$.\qed
\end{prop}
\begin{proof}
Let $1\lle i\lle r-1$. Suppose $\bm y$ is fertile in the $i$-th direction and let $\bm y^{[i]}=(y_1,\dots,\tl y_i,\dots,y_r)$ be the immediate descendant in the $i$-th direction, then we have that $\bm y^A$ is fertile in the $i$-th and $(2r-i)$-th directions. Performing the $\gl_{2m|2n}$ reproduction procedures in the $i$-th and $(2r-i)$-th directions, we can choose $\tilde y_{2r-i}=\tilde y_i$. Then 
\[
((\bm y^A)^{[i]})^{[2r-i]}=(y_1,\dots,\tl y_i,\dots,y_r,\dots,\tl y_i,\dots,y_1)=(\bm y^{[i]})^A.
\]
In addition, we have $(\bm s^{[i]})^A=((\bm s^A)^{[i]})^{[2r-i]}$. We also have
\[
(\bm y^A)^{[r]}=(y_1,\dots, y_{r-1},\tl y_r, y_{r-1},\dots,y_1)=(\bm y^{[r]})^A.
\]
Hence the map is well-defined. Clearly, the map is injective, completing the proof of the first statement. The second statement is clear from Theorem \ref{thm:A-inv} and the first statement.
\end{proof}

\subsection{Reproduction procedure for $\osp_{2m|2n}$}\label{sec D}
In this section, we apply Theorem \ref{thm:general} to the case of $\g=\osp_{2m|2n}$. 

Let $\tlbms \in \wt S^0_{m|n}$ be an extended parity sequence. 
Recall the parity sequences $\bm s^{[i]}=(s_1,\dots,s_{i+1},s_i,\dots,s_r)$, $1\lle i\lle r-1$, $\bm s^{[r]}=\bm s^{[r-1]}$, and the extended parity sequences $\tlbms^{[i]}$, $1\lle i\lle r$, and $\tlbms^{[f]}$, see \eqref{eq:tlbms}.

Let $C^{\bm s}$ be the Cartan matrix of $\g$ associated with the parity sequence $\bm s$, then, as in the case of $\glMN$, direct computations show that $(C^{\bm s})^{[i]}$ is exactly the Cartan matrix $C^{\bm s^{[i]}}$ of $\g$ associated with the parity sequence $\bm s^{[i]}$ in all cases except for the case of $i=r$, $(s_{r-1},s_r)=(-1,1)$ when  $(C^{\bm s})^{[i]}$ is obtained from $C^{\bm s^{[i]}}$  by swapping two last rows and two last columns.

The new cases, comparing to the $\glMN$, are the fermionic reproduction procedure in the $(r-1)$-st and $r$-th direction.
Here we use the first four rows of Table \ref{table:1} with $q_1=-1$, $q_2=1,2$, $q_4=1,2$, $q_3=0,2$, and sometimes we multiply a row corresponding to a root of length zero by a minus sign, see Remark \ref{DC remark}. In particular, the fermionic reproduction procedure in the $(r-1)$-st direction
may change Dynkin diagram (ignoring the parity of each node corresponding to simple roots). The new feature here is that if $\bm s$ is of type D, then the reproduction procedure in the $r$-th direction also changes the binary choice $\kappa(\tlbms)$. 

Let $\bla=(\la_1,\dots,\la_p)$ be a sequence of dominant integral $\osp_{2m|2n}$-weights, $\bm z=(z_1,\dots,z_p)$ a sequence of pairwise distinct complex numbers. Let $\bm P^{\tlbms}=(P_1^{\tlbms},\dots,P_{r}^{\tlbms})$ and $\bm T^{\tlbms}=(T_1^{\tlbms},\dots,T_{r}^{\tlbms})$ be the sequences of rational functions associated with $\bla$, $\bm z$, and $\tlbms$, see \eqref{eq:p-polynomials} and \eqref{eq:T-polynomials}. Note that in this case, $(C^{\tlbms},\bm T^{\tlbms})$ is admissible for all $\tlbms\in \wt S^0_{m|n}$.

We have $T_i^{\tlbms}=P_i^{\tlbms}(P_{i+1}^{\tlbms})^{-s_is_{i+1}}$ for $1\lle i\lle r-1$. Moreover, $T_r^{\tlbms}=P_r^{\tlbms}$ if $\bm s$ is of type C and $T_r^{\tlbms}=P_{r-1}^{\tlbms}(P_r^{\tlbms})^{s_{r-1}}$ if $\bm s$ is of type D.

Denote $\pi(T_i^{\tlbms})$ by $\pi_i^{\tlbms}$. 

Using Lemma \ref{lem: weight change}, we write $\bm P^{\tlbms^{[i]}}$ in terms of $\bm P^{\tlbms}$. For $1\lle i\lle  r-1$, if $s_i=s_{i+1}$, then $\bm P^{\tlbms^{[i]}}=\bm P^{\tlbms}$ and if $s_i\ne s_{i+1}$, then
\beq\label{eq:P-change-2}
\bm P^{\tlbms^{[i]}}=(P_1^{\tlbms},\dots,P_{i-1}^{\tlbms},P_{i+1}^{\tlbms}\pi_i^{\tlbms},P_i^{\tlbms}(\pi_i^{\tlbms})^{-1},P_{i+2}^{\tlbms},\dots,P_{r}^{\tlbms}).
\eeq
For $i=r$, we have
$$
\bm P^{\tlbms^{[r]}}=\begin{cases}
(P_1^{\tlbms},\dots,P_{r-2}^{\tlbms},(P_{r}^{\tlbms})^{-1}\pi_r^{\tlbms},P_{r-1}^{\tlbms}(\pi_r^{\tlbms})^{-1}),\qquad &\text{ if }(s_{r-1},s_r)=(-1,1),\\
\bm P^{\tlbms},\quad &\text{ otherwise }.
\end{cases}
$$
The formula for $\bm P^{\tlbms^{[r]}}$ in the case of $(s_{r-1},s_r)=(-1,1)$ corresponds to the sequence of weights $\bla$ and the extended parity $(\tlbms^{[f]})^{[r-1]}$, see \eqref{eq:r-change}. 

Now it is straightforward to check that $(\bm T^{\tlbms})^{[i]}$ obtained from Table \ref{table:1} coincides with the sequence $\bm T^{\tlbms^{[i]}}$ of rational functions associated with $\bla$, $\bm z$, and $\tlbms^{[i]}$.

Then in the case of $\osp_{2m|2n}$, Theorem \ref{thm:general} and the reproduction procedure take the following form. Set $y_0(x)=1$.
\begin{thm}\label{thm:D-rep}
Let $\bm y=(y_1,\dots,y_{r})$ be a sequence of polynomials generic with respect to $(C^{\bm s},\bm T^{\tlbms})$.
\begin{enumerate}
    \item The sequence $\bm y$ represents a solution of the Bethe ansatz equation \eqref{eq:BAE-poly-T} associated with $(C^{\bm s},\bm T^{\tlbms})$ and an extended parity sequence $\tlbms$ of type C, if and only if there exists a polynomial $\tl y_i$ satisfying \eqref{eq:A-bos}, \eqref{eq:A-fer}, for each $1\lle i\lle r-2$, and
\begin{align*}
     \begin{cases}
    \Wr(y_{r-1},\tl y_{r-1})=P_{r-1}^{\tlbms}/P_{r}^{\tlbms}y_{r-2}y_{r}^2,\\
    \Wr(y_{r},\tl y_{r})=P_{r}^{\tlbms}y_{r-1},
    \end{cases}
    \qquad \qquad\qquad \quad \text{ if }(s_{r-1},s_r)=(-1,-1),
\end{align*}
\begin{align*}
    & \begin{cases}
    y_{r-1}\tl y_{r-1}=\ln'\Big(\dfrac{P_{r-1}^{\tlbms}P_{r}^{\tlbms}y_{r-2}}{y_r^2}\Big)\pi_{r-1}^{\tlbms}y_{r-2}y_{r},\\
    \Wr(y_{r},\tl y_{r})=P_{r}^{\tlbms}y_{r-1},
    \end{cases}
    &\text{ if }(s_{r-1},s_r)=(1,-1).
\end{align*}
    \item The sequence $\bm y$ represents a solution of the Bethe ansatz equation \eqref{eq:BAE-poly-T} associated with $(C^{\bm s},\bm T^{\tlbms})$ and an extended parity sequence $\tlbms$ of type D, if and only if there exists a polynomial $\tl y_i$ satisfying \eqref{eq:A-bos}, \eqref{eq:A-fer}, for each $1\lle i\lle r-3$, and
\begin{align*}
    & \begin{cases}
    \Wr(y_{r-2},\tl y_{r-2})=P_{r-2}^{\tlbms}/P_{r-1}^{\tlbms}y_{r-3}y_{r-1}y_r,\\
    \Wr(y_{r-1},\tl y_{r-1})=P_{r-1}^{\tlbms}/P_{r}^{\tlbms}y_{r-2},\\
    \Wr(y_{r},\tl y_{r})=P_{r-1}^{\tlbms}P_{r}^{\tlbms}y_{r-2},
    \end{cases}
    &\text{ if }(s_{r-1},s_r)=(1,1,1),
\end{align*}
\begin{align*}
     \begin{cases}
    y_{r-2}\tl y_{r-2}=\ln'\Big(\dfrac{P_{r-2}^{\tlbms}P_{r-1}^{\tlbms}y_{r-3}}{y_{r-1}y_r}\Big)\pi_{r-2}^{\tlbms}y_{r-3}y_{r-1}y_r,\\
    \Wr(y_{r-1},\tl y_{r-1})=P_{r-1}^{\tlbms}/P_{r}^{\tlbms}y_{r-2},\\
    \Wr(y_{r},\tl y_{r})=P_{r-1}^{\tlbms}P_{r}^{\tlbms}y_{r-2},
    \end{cases}
    \text{ if }(s_{r-1},s_r)=(-1,1,1),
\end{align*}
\begin{align*}
    & \begin{cases}
    \Wr(y_{r-2},\tl y_{r-2})=P_{r-2}^{\tlbms}/P_{r-1}^{\tlbms}y_{r-3}y_{r-1}y_r,\\
    y_{r-1}\tl y_{r-1}=\ln'\Big(\dfrac{P_{r-1}^{\tlbms}P_{r}^{\tlbms}y_{r-2}}{y_r^2}\Big)\pi_{r-1}^{\tlbms}y_{r-2}y_r,\\
    y_{r}\tl y_{r}=\ln'\Big(\dfrac{P_{r-1}^{\tlbms}y_{r-2}}{P_{r}^{\tlbms}y_{r-1}^2}\Big)\pi_{r}^{\tlbms}y_{r-2}y_{r-1},
    \end{cases}
    &\text{ if }(s_{r-1},s_r)=(-1,-1,1),
\end{align*}
\begin{align*}
 \begin{cases}
    y_{r-2}\tl y_{r-2}=\ln'\Big(\dfrac{P_{r-2}^{\tlbms}P_{r-1}^  {\tlbms}y_{r-3}}{y_{r-1}y_r}\Big)\pi_{r-2}^{\tlbms}y_{r-3}y_{r-1}y_r,\\
    y_{r-1}\tl y_{r-1}=\ln'\Big(\dfrac{P_{r-1}^{\tlbms}P_{r}^{\tlbms}y_{r-2}}{y_r^2}\Big)\pi_{r-1}^{\tlbms}y_{r-2}y_r,\\
    y_{r}\tl y_{r}=\ln'\Big(\dfrac{P_{r-1}^{\tlbms}y_{r-2}}{P_{r}^{\tlbms}y_{r-1}^2}\Big)\pi_{r}^{\tlbms}y_{r-2}y_{r-1},
    \end{cases}
    \quad \ \ \text{ if }(s_{r-1},s_r)=(1,-1,1).
\end{align*}
    \item Let 
    \beq\label{change label} \bm y^{[i]}=\begin{cases}(y_1,\dots,y_{r-2}, \tl y_{r},y_{r-1}), & {\text{if}}\ i=r,\ (s_{r-1},s_r)=(-1,1), \\
    (y_1,\dots,\tl y_i,\dots,y_{r}), & {\text{otherwise.}}\end{cases}\eeq
    If $\bm y^{[i]}$ 
    is generic with respect to $(C^{\bm s^{[i]}},\bm T^{\tlbms^{[i]}})$, then $\bm y^{[i]}$ represents a solution of the Bethe ansatz equation \eqref{eq:BAE-poly-T} associated with $(C^{\bm s^{[i]}},\bm T^{\tlbms^{[i]}})$. \qed
\end{enumerate}
\end{thm}

Similar to Theorems \ref{thm:A-rep}, \ref{thm:B-rep}, in each case one can remove all $P_j^{\tlbms}$ and write instead a single $T_i^{\tlbms}$. 

We now discuss the role of the choice $\kappa$ in an extended parity sequence $\tlbms$ of type D. The change $\kappa\to -\kappa$ is resulting in the changes:
$P_r^{\tlbms}\to (P_r^{\tlbms})^{-1}$, $T_{r-1}^{\tlbms}\leftrightarrow T_r^{\tlbms}$,  $y_{r-1}\leftrightarrow y_r$. In other words, this change represents the non-trivial involution of the Dynkin diagram exchanging the labels of nodes $r-1$ and $r$. In particular, it does not change the solution of the Bethe ansatz equation. To accommodate for this, it is convenient to include such a change into reproduction procedure. For that, we introduce a {\it fake reproduction procedure} as follows. 

Recall that if $s_r=1$, we set $\tlbms^{[f]}=(\bm s;-\kappa(\tlbms))$, 
$$
\bs y^{[f]}=(y_1,\dots,y_{r-2},y_r,y_{r-1}).
$$
Then we also have
\begin{lem}\label{lem:fake}
Let $\bm y=(y_1,\dots,y_{r})$ represent a solution of the Bethe ansatz equation \eqref{eq:BAE-poly-T} associated with $(C^{\bm s},\bm T^{\tlbms})$. Assume $s_r=1$. Then 
$\bm y^{[f]}$ represents a solution of the Bethe ansatz equation \eqref{eq:BAE-poly-T} associated with $(C^{\bm s},\bm T^{\tlbms^{[f]}})$. \qed
\end{lem}

Due to Lemma \ref{lem:fake},  if $s_r=1$, we have
\beq\label{eq:r and fake}
(\bm y^{[r]},\tlbms^{[r]})=\begin{cases}
\big((\bm y^{[f]})^{[r-1]},(\tlbms^{[f]})^{[r-1]}\big), & \text{ if }s_{r-1}=-1,\\
\big(((\bm y^{[f]})^{[r-1]})^{[f]},((\tlbms^{[f]})^{[r-1]})^{[f]}\big), & \text{ otherwise}.
\end{cases}
\eeq
In type D, there is no natural order for the last two simple roots - and the fake reproduction procedure just exchanges the labeling of these two simple roots. However, in type C we do have a natural labeling which may be different from the labeling chosen for type D. This is taken care by the change of order of components of $\bs y$ in \eqref{change label}.
\medskip

Starting from a pair
$(\bm y,\tlbms)$, consisting of a solution of the Bethe ansatz equation and the corresponding extended parity sequence, 
we produce a collection of similar pairs by repeatedly applying all possible reproduction procedures (including the fake one). We denote this collection by $\mathscr P_{\bm y,\tlbms}\subset (\mathbb P(\C[x]))^{r}\times \tilde S_{m|n}^0$ and call it the $\osp_{2m|2n}$ \emph{population of solutions of the Bethe ansatz equation} originated from $(\bm y,\tlbms)$.

\subsection{The differential operator of an $\osp_{2m|2n}$ population.}\label{sec D2}
As in the case of $\osp_{2m+1|2n}$, we include the $\osp_{2m|2n}$ population into a population of type $A$, in this case, associated to $\gl_{2m|2n+1}$. We do it under the additional assumption that all weights $\la_i$ are dominant integral of the same kind. 

Let all weights $\la_i$ be dominant integral of the first kind. Under this assumption all $P_i^{\tlbms}$ are polynomials provided $\kappa(\tlbms)=1$. 
For an extended parity sequence $\tlbms\in\wt S^0_{m|n}$, define
$$
\bm s^A=(s_1^A,\dots,s_{2r+1}^A)=(s_1,\dots,s_r,-1,s_r,\dots,s_1)\in S_{2m|2n+1}.
$$ 

Let $\bm y=(y_1,\dots,y_{r})$ be a sequence of polynomials. Define the sequence of polynomials,
\beq\label{eq:y-A-D}
\bm y^{\tlbms,A}=\begin{cases}
(y_1,\dots,y_{r-2},y_{r-1}y_r,y_r^2,y_r^2,y_{r-1}y_r,y_{r-2},\dots,y_1),&\text{ if }\kappa(\tlbms)=1,\ s_r=1,\\
(y_1,\dots,y_{r-2},y_{r-1}y_r,y_{r-1}^2,y_{r-1}^2,y_{r-1}y_r,y_{r-2},\dots,y_1),&\text{ if }\kappa(\tlbms)=-1,\ s_r=1,\\
(y_1,\dots,y_{r-1},y_r^2,y_r^2,y_{r-1},\dots,y_1), & \text{ otherwise}.
\end{cases}
\eeq

Let $\bla$, $\bm P^{\tlbms}$, and $\bm T^{\tlbms}$ be as above. 

Define $\gl_{2m|2n+1}$-weights $\la_i^A$, $1\lle i\lle p$, by $((\la_i^A)^{\bm s^A},\ve^{\bm s^A}_{r+1})=0$ and
\[
\big((\la_i^A)^{\bm s^A},\ve^{\bm s^A}_j\big)=-\big((\la_i^A)^{\bm s^A},\ve^{\bm s^A}_{2r+2-j}\big)=
(\la_i^{\tlbms},\ve^{\tlbms}_j),\qquad 1\lle j\lle r.
\]
The weights $\la_i^A$ apriori depend on $\tlbms$ but by Lemma \ref{lem: weight change} they depend on $\kappa(\tlbms)$ only. To make this dependence explicitly we denote the sequence  $(\la_1^A,\dots,\la_p^A)$ by 
$\bla^{\kappa,A}=(\la_1^{\kappa,A},\dots,\la_p^{\kappa,A})$.

\begin{lem}
If $\la$ is a dominant integral $\osp_{2m|2n}$-weight of the first kind (resp. of the second kind), then $\la^{1,A}$ (resp. $\la^{-1,A}$) is a dominant integral $\gl_{2m|2n+1}$-weight.
\end{lem}
\begin{proof}
One computes $\la^{1,A}$ starting from $(\la^{1,A})^{\bm s^A}$ by moving components corresponding to negative $s_i$ to the right. The rule is
described in Lemma \ref{lem: weight change}. For the extended parity sequence $\tlbms_-=(\bm s_-;1)$,  we have
$$
(\la^{1,A})^{(\bm s_-)^A}=(\mu_1,\dots,\mu_{r},0,-\mu_{r},\dots,-\mu_1),
$$
where we have $\mu_1\gge\cdots\gge\mu_m$ and  $\mu_{m+1}\gge\cdots\gge \mu_{r}$.
Clearly, moving the first $m$ components and the zero to the right creates a dominant integral $\gl_{2m|2n+1}$-weight.
\end{proof}

We say that an $\osp_{2m|2n}$-weight $\la$ of the first kind (resp. of the second kind) is A-\textit{typical} if $\la^{1,A}$ (resp. $\la^{-1,A}$) is a typical $\gl_{2m|2n+1}$-weight. We say that $\bla$ is A-\textit{typical} if all $\la_i$ are dominant integral weights of the same kind and if at least one of $\la_i$ is A-typical.

\medskip

Define a sequence of rational functions $\bm P^{\tlbms,A}$ by
\beq\label{eq:PA-C}
\bm P^{\tlbms,A}=(P_1^{\tlbms,A},\dots,P_{2r+1}^{\tlbms,A}),\quad P^{\tlbms,A}_i=(P^{\tlbms,A}_{2r+2-i})^{-1}=\begin{cases} 
1,&\text{ if }i=r+1,\\
P^{\tlbms}_i,&\text{ otherwise},
\end{cases}
\eeq
and a sequence of rational functions $\bm T^{\tlbms,A}$ by
\beq\label{eq:TA-C}
\bm T^{\tlbms,A}=(T_1^{\tlbms,A},\dots,T_{2r}^{\tlbms,A}),
\quad T_i^{\tlbms,A}=(T_{2r+1-i}^{\tlbms,A})^{s_i^As_{i+1}^A} =\begin{cases}P_r^{\tlbms},  &\text{ if }i=r \text{ and }s_r=1,\\
T_i^{\tlbms}, & \text{otherwise},
\end{cases}
\eeq
where $1\lle i\lle r$. Note that $\bm P^{\tlbms,A}$ and $\bm T^{\tlbms,A}$ are exactly sequences of rational functions associated to $\bla^{\kappa(\tlbms),A}$, $\bm z$, and $\bm s^A$, see \eqref{eq:p-polynomials} and \eqref{eq:T-polynomials}. Set $\pi_{i}^{\tlbms,A}:=\pi( T_i^{\tlbms,A})$, $1\lle i\lle 2r$.

We say that the sequence $\bs T^{\tlbms}$ is A-{\it typical} if the sequence of $\gl_{2m|2n+1}$-weights $\bla^{\kappa(\tlbms),A}$ is A-typical.

Let $\tlbms$ be an extended parity sequence of type C with $\kappa(\tlbms)=1$. Let $\bm y$ represent a solution of the Bethe ansatz equation as in Theorem \ref{thm:D-rep}. We shall construct a $\gl_{2m|2n+1}$ population containing $(\bm y,\tlbms)$. 

Let $\tilde y_r$ be as in \eqref{eq:wr-div}, that is $\Wr(y_r,\tilde y_r)=T_r^{\tlbms} y_{r-1}$. We choose $\tilde y_r$ in such a way that it is relatively prime to the polynomial $T_r^{\tlbms} y_{r-1}y_r$.
Set
\begin{align*}
&\bm y^{\tlbms,A}_1=(y_1,\dots,y_{r-1},y_r\tl y_r,y_r^2,y_{r-1},\dots,y_1),\\
&\bm y^{\tlbms,A}_2=(y_1,\dots,y_{r-1},y_r\tl y_r,y_r^2+c\tl y_r^2,y_{r-1},\dots,y_1).
\end{align*}
Let $C^{\bm s,A}$ be the Cartan matrix of $\gl_{2m|2n+1}$ associated with the parity sequence $\bm s^A$.
\begin{lem}\label{lem:AC}
For almost all $c\in\C$, the sequence of polynomials $\bm y^{\tlbms,A}_2$ represents a solution of the Bethe ansatz equation associated with $(C^{\bm s,A},\bm T^{\tlbms,A})$. The $\gl_{2m|2n+1}$ population, containing $(\bm y^{\tlbms,A}_2,\bm s^A)$, does not depend on $c$ and contains also the sequence $(\bm y^{\tlbms,A},\bm s^A)$.
\end{lem}
\begin{proof}
The lemma for the case of $(s_{r-1},s_r)=(-1,-1)$ follows from \cite[Lemma 7.6]{MV04}. We only need to show the case of $(s_{r-1},s_r)=(1,-1)$. Since $\bm y$ is generic with respect to $(C^{\bm s},\bm T^{\tlbms})$, clearly $\bm y_2^{\tlbms,A}$ is generic with respect to $(C^{\bm s,A},\bm T^{\tlbms,A})$ for almost all $c\in \C$. It suffices to show that $\bm y_2^{\tlbms,A}$ is fertile in all directions with respect to $(C^{\bm s,A},\bm T^{\tlbms,A})$.

By Theorem \ref{thm:D-rep}, \eqref{eq:T-polynomials}, and \eqref{eq:PA-C}, we have
\[
\Wr(y_r^2,y_r\tl y_r)=T_r^{\tlbms,A} y_{r-1}y_r^2,\qquad y_{r-1}(\tl y_{r-1}y_r)=\ln'\Big(\dfrac{T_{r-1}^{\tlbms,A}y_{r-2}}{y_r^2}\Big)\pi_{r-1}^{\tlbms,A}y_{r-2}y_{r}^2.
\]
These equations show that $\bm y^{\tlbms,A}$ is fertile  in all directions with respect to $(C^{\bm s,A},\bm T^{\tlbms,A})$. Note that 
\[
\Wr(y_r^2,y_r^2+c\tl y_r^2)=2cT_r^{\tlbms,A}y_{r-1}(y_r\tl y_{r}),
\]
and we conclude that $\bm y_1^{\tlbms,A}$ is fertile in all directions except possibly the $(r-1)$-st direction with respect to $(C^{\bm s,A},\bm T^{\tlbms,A})$. We also have
\begin{align*}
    \ln'\Big(\frac{T_{r-1}^{\tlbms,A}y_{r-2}}{y_r\tl y_r}\Big)\pi_{r-1}^{\tlbms,A}y_{r-2}y_r\tl y_r-&\ \tl y_r \ln'\Big(\frac{T_{r-1}^{\tlbms,A}y_{r-2}}{y_r^2}\Big)\pi_{r-1}^{\tlbms,A}y_{r-2}y_r\\=&\ -\pi_{r-1}^{\tlbms,A}y_{r-2}\Wr(y_r,\tl y_r)=-\pi_{r-1}^{\tlbms,A}y_{r-2}T_r^{\tlbms} y_{r-1}.
\end{align*}
Since $y_{r-1}$ divides the polynomial $\ln'\Big(\frac{T_{r-1}^{\tlbms,A}y_{r-2}}{y_r^2}\Big)\pi_{r-1}^{\tlbms,A}y_{r-2}y_r$, $y_{r-1}$ divides $\ln'\Big(\frac{T_{r-1}^{\tlbms,A}y_{r-2}}{y_r\tl y_r}\Big)\pi_{r-1}^{\tlbms,A}y_{r-2}y_r\tl y_r$ as well. This shows $\bm y_1^{\tlbms,A}$  is fertile with respect to $(C^{\bm s,A},\bm T^{\tlbms,A})$ in the $(r-1)$-st direction. In particular, $\bm y_2^{\tlbms,A}$ is fertile in all directions except possibly in the $r$-th and $(r+2)$-nd directions with respect to $(C^{\bm s,A},\bm T^{\tlbms,A})$. We show that $\bm y_2^{\tlbms,A}$
is fertile in these directions too. 

The fact that $\bm y_2^{\tlbms,A}$
is fertile in the $r$-th direction follows from the equality
\[
\Wr(y_r\tl y_r, -y_r^2+c\tl y_r^2)=T_{r}^{\tlbms,A}y_{r-1}(y_r^2+c\tl y_r^2).
\]
Due to the equality
\[
2\ln'\Big(\frac{T_{r-1}^{\tlbms,A}y_{r-2}}{y_r\tl y_r}\Big)\pi_{r-1}^{\tlbms,A}y_{r-2}y_r\tl y_r= \tl y_r \ln'\Big(\frac{T_{r-1}^{\tlbms,A}y_{r-2}}{y_r^2}\Big)\pi_{r-1}^{\tlbms,A}y_{r-2}y_r+  y_r \ln'\Big(\frac{T_{r-1}^{\tlbms,A}y_{r-2}}{\tl y_r^2}\Big)\pi_{r-1}^{\tlbms,A}y_{r-2}\tl y_r
\]
and the fact that $y_r$ and $y_{r-1}$ are relatively prime, we deduce that the polynomial $\ln'\Big(\frac{T_{r-1}^{\tlbms,A}y_{r-2}}{\tl y_r^2}\Big)\pi_{r-1}^{\tlbms,A}y_{r-2}\tl y_r$ is divisible by $y_{r-1}$. 

Finally, the fact that $\bm y_2^{\tlbms,A}$
is fertile in the $(r+2)$-nd direction follows from $T_{r+2}^{\tlbms,A}=(T_{r-1}^{\tlbms,A})^{-1}$ and the equality
\begin{align*}
-\ln'\Big(\frac{T_{r+2}^{\tlbms,A}(y_r^2+c\tl y_r^2)}{y_{r-2}}\Big)&\pi_{r-1}^{\tlbms,A}y_{r-2}(y_r^2+c\tl y_r^2)\\&\ =  \ln'\Big(\frac{T_{r-1}^{\tlbms,A}y_{r-2}}{y_r^2}\Big)\pi_{r-1}^{\tlbms,A}y_{r-2}y_r^2+c \ln'\Big(\frac{T_{r-1}^{\tlbms,A}y_{r-2}}{\tl y_r^2}\Big)\pi_{r-1}^{\tlbms,A}y_{r-2}\tl y_r^2.\qedhere
\end{align*}
\end{proof}

Let $\bm y=(y_1,\dots,y_{r})$ be a sequence of polynomials and $\tlbms$ an arbitrary extended parity sequence. Define a rational pseudo-differential operator $\mathscr R^{2m|2n}_{\bm y,\tlbms}$ by
\beq\label{eq:D-oper}
\mathscr R^{2m|2n}_{\bm y,\tlbms}=
\mathscr R_{\bm y^{\tlbms,A},\bm s^A}=\mathop{\overrightarrow\prod}\limits_{1\lle i\lle 2r+1}\Big(\pa-s_i^{A}\ln'\dfrac{P_i^{\tlbms,{A}}y^{\tlbms,A}_{i-1}}{y_i^{\tlbms,A}}\Big)^{s_i^{A}},
\eeq
where $y_{0}^{\tlbms,A}=y_{2r+1}^{\tlbms,A}=1$. Clearly, $\mathscr R^{2m|2n}_{\bm y,\tlbms}$ is symmetric.

\begin{lem}\label{lem fake-inv}
Suppose $s_r=1$. We have $\mathscr R^{2m|2n}_{\bm y^{[f]},\tlbms^{[f]}}=\mathscr R^{2m|2n}_{\bm y,\tlbms}$. In other words, the rational pseudo-differential operator $\mathscr R^{2m|2n}_{\bm y,\tlbms}$ is invariant under the fake reproduction procedure.
\end{lem}
\begin{proof}
The first and last $r-1$ factors of $\mathscr R^{2m|2n}_{\bm y^{[f]},\tlbms^{[f]}}$ and $\mathscr R^{2m|2n}_{\bm y,\tlbms}$ coincide. Hence it suffices to check that the products of the middle 3 factors
\[
\Big(\pa+\ln'\frac{P_r^{\tlbms}y_{r-1}}{y_{r}}\Big)\pa^{-1}\Big(\pa-\ln'\frac{P_r^{\tlbms}y_{r-1}}{y_r}\Big),
\quad
\Big(\pa-\ln'\frac{P_r^{\tlbms}y_{r-1}}{y_r}\Big)\pa^{-1}\Big(\pa+\ln'\frac{P_r^{\tlbms}y_{r-1}}{y_{r}}\Big)
\]
are the same. This follows from Lemma \ref{lem:swit-chpm}.
\end{proof}

Let $\tlbms_0$ be an extended parity sequence such that $\kappa(\tlbms_0)=1$. Suppose that the sequence $\bm y_0$ represents a solution of the $\osp_{2m|2n}$ Bethe ansatz equation associated with $(C^{\bs s_0},\bm T^{\tlbms_0})$. Let $\mathscr P_{\bm y_0,\tlbms_0}$ be the $\osp_{2m|2n}$ population originated from $(\bm y_0,\tlbms_0)$. Let $\mathscr P^A_{\bm y_0,\tlbms_0}$ be the $\gl_{2m|2n+1}$ population originated from $(\bm y_0^{\tlbms_0,A},\bm s_0^A)$, cf. Lemma \ref{lem:AC}. 

\begin{prop}\label{prop:D-inv} 
The rational pseudo-differential operator $\mathscr R^{2m|2n}_{\bm y,\tlbms}$ is independent of the choice of $(\bm y,\tlbms)$ in $\mathscr P_{\bm y_0,\tlbms_0}$. Assume that $\bs T^{\tlbms}$ is A-typical, then there exists a natural map from $\mathscr P_{\bm y_0,\tlbms_0}$ to the set of symmetric complete factorizations of $\mathscr R^{2m|2n}_{\bm y,\tlbms}$ given by \eqref{eq:D-oper}.\qed
\end{prop}
\begin{proof}
First, we show that the rational pseudo-differential operator $\mathscr R^{2m|2n}_{\bm y,\tlbms}(\bm y)$ does not change under the reproduction procedure. Since the reproduction procedure in the $r$-th direction is expressed in terms of reproduction procedure in the $(r-1)$-th direction and the fake reproduction procedure, see \eqref{eq:r and fake}, it is enough to show it for directions $1,\dots,r-1$ and for the fake reproduction procedure. 

The case of fake reproduction procedure follows from Lemma \ref{lem fake-inv}. Reproductions procedures in directions $1,2,\dots$, $r-2$ are treated as in Proposition \ref{prop:B-inv}. For reproduction procedure in the  $(r-1)$-st direction, we have several cases. We discuss in detail the case of $(s_{r-1},s_r)=(1,-1)$. The other cases are similar. 

Since $\bm y$ is $\osp_{2m|2n}$ fertile in the $(r-1)$-st direction, by definition (cf. Theorem \ref{thm:D-rep}), there exists a polynomial $\tl y_{r-1}$ satisfying
\[
y_{r-1}\tl y_{r-1}=\ln'\Big(\dfrac{P_{r-1}^{\tlbms}P_{r}^{\tlbms}y_{r-2}}{y_r^2}\Big)\pi_{r-1}^{\tlbms}y_{r-2}y_{r}\Longrightarrow y_{r-1}(\tl y_{r-1}y_r)=\ln'\Big(\dfrac{P_{r-1}^{\tlbms,A}P_{r}^{\tlbms,A}y_{r-2}}{y_r^2}\Big)\pi_{r-1}^{\tlbms,A}y_{r-2}y_{r}^2.
\]
Hence the sequence $\bm y^{\tlbms,A}=(y_1,\dots,y_{r-1},y_r^2,y_r^2,y_{r-1},\dots,y_1)$ is $\gl_{2m|2n+1}$ fertile in the $(r-1)$-st direction. In particular, we have 
$$(\bm y^{\tlbms,A})^{[r-1]}=(y_1,\dots,y_{r-2},\tl y_{r-1}y_r,y_r^2,y_r^2,y_{r-1},\dots,y_1).$$ Note that by \eqref{eq:PA-C} and \eqref{eq:TA-C}, we also have
\[
y_{r-1}(\tl y_{r-1}y_r)=-\ln'\Big(\dfrac{P_{r+2}^{\tlbms,A}P_{r+3}^{\tlbms,A}y_r^2}{y_{r-2}}\Big)\pi_{r+2}^{\tlbms,A}y_{r}^2y_{r-2}.
\]
Therefore $(\bm y^{\tlbms,A})^{[r-1]}$ is $\gl_{2m|2n+1}$-fertile in the $(r+2)$-nd direction and 
\[
((\bm y^{\tlbms,A})^{[r-1]})^{[r+2]}=(y_1,\dots,y_{r-2},\tl y_{r-1}y_r,y_r^2,y_r^2,\tl y_{r-1}y_r,y_{r-2},\dots,y_1)=(\bm y^{[r-1]})^{\tlbms^{[r-1]},A}.
\]
Thus $((\bm y^{[r-1]})^{\tlbms^{[r-1]},A},(\bm s^{[r-1]})^A)$ is in the $\gl_{2m|2n+1}$ population containing $(\bm y^{\tlbms,A},\bm s^A)$. By Theorem \ref{thm:A-inv}, we conclude that the rational pseudo-differential operator does not change.

If $\bs T^{\tlbms}$ is A-typical, then $\mathscr R^{2m|2n}_{\bm y,\tlbms}$ is a $(2m|2n+1)$-rational pseudo-differential operator. Therefore, for any 
$(\bs y,\tlbms)\in \mathscr P_{\bm y_0,\tlbms_0}$, formula \eqref{eq:D-oper} gives a symmetric complete factorization of $\mathscr R^{2m|2n}_{\bm y,\tlbms}$.
\end{proof}
      
It is clear that the map described in Proposition  \ref{prop:D-inv} is injective when restricted to points of $\mathscr P_{\bm y_0,\tlbms_0}$  with fixed value of $\kappa(\tlbms)$. We will show that it is injective in Theorem \ref{main thm}, see also Remark \ref{rem after thm}.

\section{Populations and isotropic superflag varieties}\label{sec: main}
\subsection{Superfertile solutions of the Bethe ansatz equations}\label{sec:superfertility}
Let $C$ be a generalized Cartan matrix. Let $\bm T$ be a sequence of functions such that $(C,\bm T)$ is admissible. Let $\bs y$ represent a solution of the Bethe ansatz equation associated with admissible $(C,\bm T)$. In particular, $\bs y$ is generic with respect to $(C,\bm T)$. Then by Theorem \ref{thm:general}, if a descendant $\bs y^{[i]}$ is generic with respect to $(C^{[i]},\bm T^{[i]})$ and $(C^{[i]},\bm T^{[i]})$ is admissible, we can repeat the reproduction procedure. Sometimes 
$\bs y^{[i]}$ is not generic, but still fertile and we can do the reproduction procedure. In this section we study the case when the reproduction procedure can be done indefinitely.

Following the terminology of  \cite{MV08}, we call a fertile tuple $\bs y$ \textit{superfertile} if all sequences in the population of $\bs y$ are fertile.

Note that ``super" in superfertile has nothing to do with the parity and just means ``always fertile". 

It is expected that a sequence $\bs y$ representing a solution of the Bethe ansatz equation associated with sufficiently generic weights $\bla$ and evaluation points $\bm z$ is always superfertile.  

In the even situation, a generic fertile sequence $\bs y$ is indeed always superfertile, thanks to \cite[Lemma 3.6]{MV04}. In particular, it applies to the cases of D$_m$ or $\osp_{2m|0}$, B$_m$ or  $\osp_{2m+1|0}$, C$_n$ or  $\osp_{0|2n}$, and  $\osp_{1|2n}$.

In the presence of fermionic reproduction procedures, one has to add some restrictions. 

Recall the parity sequence $\bm s_+=(1,\dots,1,-1,\dots,-1)\in S_{m|n}$. The sequence $\bs y$ representing a solution of the $\glMN$ Bethe ansatz equation associated to parity sequence $\bm s_+$ under assumption that at least one weight is typical is superfertile, \cite{HMVY19}. Note that this assumption is sufficient but not necessary.

Let $\bla=(\la_1,\dots,\la_p)$ be a sequence of dominant integral $\osp_{2m+\iota|2n}$-weights (of the first kind if $\iota=0$), $\bm z=(z_1,\dots,z_p)$ a sequence of pairwise distinct complex numbers. Recall that $\tlbms_+=(\bm s_+;1)\in\wt S_{m|n}^\iota$. Let $\bm T^{\tlbms_+}=(T_1^{\tlbms_+},\dots,T_{r}^{\tlbms_+})$ be the sequence of rational functions associated with $\bla$, $\bm z$, and $\tlbms_+$. 

\begin{conj}\label{conj suprefertile}
Let $\bs y$ represent a solution of $\osp_{2m+\iota|2n}$ Bethe ansatz equation associated with $(C^{\bm s_+},\bm T^{\tlbms_+})$. Assume that all roots of $y_i$ are of order one and do not vanish at $z_j$, $j=1,\dots, p$. Assume that all weights $\la_j$ are $A$-typical. Then $\bs y$ is superfertile.\qed
\end{conj}

It is likely that the assumptions of this conjecture could be relaxed further.

\begin{prop}\label{conj true}
Conjecture \ref{conj suprefertile} holds for the cases  $\osp_{3|2n}$ and $\osp_{2m+1|2}$.
\end{prop}

Proposition \ref{conj true} is proved in Section \ref{sec proof of prop}.

\subsection{More on superflags}
Let $W$ be a self-dual superspace of dimension $M+N$. Let $\mathscr F=\{F_1\subset \cdots \subset F_{M+N}=W\}\in\mathscr F(W)$ be a full superflag. Let $w_1,w_2,\dots,w_{M+N}$ be a basis of $W$ which generates $\mathscr F$. In other words, $w_1,\dots,w_i$ is a basis of $F_i$ for each $1\lle i\lle M+N$.

Let $\mathscr F'\in \mathscr F(W)$ be another superflag. Then there exists a unique basis $w_1',\dots,w'_{M+N}$ of $W$, $w_i'=\sum_ja_{ij}w_j$, which generates $\mathscr F'$ and such that $a_{ij}=0$ whenever there exists $k<i$ with the property $a_{kj}\neq 0$, $a_{ks}=0$, for each $j+1\lle s\lle M+N$, and such that $a_{ij}=1$ whenever $a_{ij}\neq 0$, $a_{is}=0$,  for each $j+1\lle s\lle M+N$.

Define a permutation $\sigma\in\fkS_{M+N}$ by setting $\sigma(i)$ to be the largest $j$ such that $a_{ij}\neq 0$. 

Clearly, $\sigma$ is a well-defined permutation which depends only on superflags $\mathscr F$, $\mathscr F'$, and not on the choice of the generating basis $\{w_1,\dots,w_{M+N}\}$.

Thus, a fixed superflag $\mathscr F$ defines a cell decomposition of the set of superflags $\mathscr F(W)$
$$
\mathscr F(W)=\bigsqcup_{\sigma\in\fkS_{M+N}} \mathscr F_\sigma(W),
$$
where $\mathscr F_\sigma(W)$ is the set of superflags corresponding to the permutation $\sigma$. We also have a
cell decomposition of the isotropic superflags, 
$$
\mathscr F^\perp(W)=\bigsqcup_{\sigma\in\fkS_{M+N}} \mathscr F^\perp_\sigma(W),
$$
We call these decompositions the {\it Bruhat cell decompositions}.

Note that if $\mathscr F$ is an isotropic superflag then $\mathscr F^\perp_\sigma(W)=\emptyset$ unless $\sigma(i)+\sigma(M+N+1-i)=M+N+1$ for all $1\lle i\lle M+N$. 

In the case when $W$ is odd dimensional, we will use the following decompositions,
\beq\label{eq:pm-decom}
\mathscr F^\perp(W)= \mathscr F^\perp_{+}(W) \bigsqcup \mathscr F^\perp_{-}(W), \qquad \mathscr F^\perp_{\pm }(W)=\bigsqcup_{\substack{ \sigma\in\fkS_{M+N}, \\  \mathrm{sgn}(\sigma)=\pm 1}} \mathscr F^\perp_\sigma(W).
\eeq
In this case, $\mathscr F^\perp_{\pm}(W)$ are closed in  $\mathscr F^\perp(W)$.

\subsection{Reproduction procedure and superflags} Let $\g=\osp_{2m+\iota|2n}$. Let $\bla$ be an $A$-typical sequence of dominant integral $\osp_{2m+\iota|2n}$-weights, $\bm z$ a sequence of pairwise distinct complex numbers. Let $\tlbms_0$ be an extended parity sequence in $\wt S_{m|n}^\iota$. Let $C^{\bm s_0}$ be the Cartan matrix of $\g$ associated to the parity sequence $\bm s_0$. Let $\bm T^{\tlbms_0}$ be the sequence of rational functions associated with $\bla$, $\bm z$, and $\tlbms_0$, see \eqref{eq:T-polynomials}.

Let $\bs y_0$ be a sequence of polynomials which represents a solution of the Bethe ansatz equation \eqref{eq:BAE-poly-T-nos} associated to $(C^{\bm s_0},\bs T^{\tlbms_0})$. Let $\mathscr P=\mathscr P_{\bm y_0,\tlbms_0}$ be the population of solutions of the Bethe ansatz equation originated from $(\bs y_0,\tlbms_0)$

Then $\mc R=\mathscr R^{2m+\iota|2n}_{\bm y_0,\tlbms_0}$ is a symmetric $(2m|2n+1-\iota)$-rational pseudo-differential operator. Let $W$ be the superkernel of $\mc R$. Then $W$ is a self-dual superspace of functions, see Lemma \ref{lem:super-self-dual}.

The sequence $\bs y$ determines a symmetric complete $\bs s_0^A$-factorization of $\mc R$, see \eqref{eq:B-oper} and \eqref{eq:D-oper}. By Proposition \ref{prop:bij-f-iso}, symmetric complete factorizations of $\mc R$ are in a bijection with isotropic superflags in $W$. We have $\dim W=2m+2n+1-\iota$ and we denote it by $r^{A}$. Let $\mathscr F_0\in\mathscr F(W)$ be the superflag corresponding to the symmetric complete factorization of $\mc R$ defined by $\bs y_0$.

Any other pair $(\bm y,\tlbms)\in\mathscr P$ in the population of solutions of the Bethe ansatz equation originated from $(\bs y_0,\tlbms_0)$ gives a symmetric complete $\bm s^A$-factorization of  $\mc R$ and, therefore, corresponds to an isotropic superflag of $W$ which we denote by $\mathscr F_{\bm y,\tlbms}$. We describe the change of the superflag during the reproduction procedure.

Let $\{w_1,\dots, w_{r^A}\}$ be a homogeneous basis of $W$ which generates the superflag $\mathscr F_{\bs y,\tlbms}$. Then $\{w_1,\dots, w_i\}$ is a basis of dimension $i$ subspace in superflag $\mathscr F_{\bs y,\tlbms}$. In particular, parity of  $w_i$ is $s^A_i$.

We have three principal cases. 

\medskip 

First, a bosonic reproduction procedure, cf.  \cite[Section 6.4]{MV04}. In this case, an immediate  descendant depends on a choice of integration constant $c\in\C$.

Let $i<r=m+n$ and $s_i=s_{i+1}$. Assume that we are not in type D with $\kappa(\tlbms)=-1$ and $i=r-1$. Then the superflag $\mathscr F_{\bs y^{[i]},\tlbms^{[i]}}$ corresponding to the immediate descendent in the $i$-th direction is generated by the basis
\beq\label{eq:rep-flag-1}
\{w_1,\dots, w_{i-1},cw_i+w_{i+1},w_{i},w_{i+2},\dots,w_{r^A-i-1}, cw_{r^A-i}+w_{r^A-i+1},w_{r^A-i},w_{r^A-i+2},\dots, w_{r^A}\}.
\eeq
In the case, $i=r-1$, $\iota=0$,  $s_r=1$ (type D), $\kappa(\tlbms)=-1$, the basis above generates the superflag $\mathscr F_{\bs y^{[r]},\tlbms^{[r]}}$. 

Let now $i=r$, $\iota=1$. Then the superflag $\mathscr F_{\bs y^{[r]},\tlbms^{[r]}}$ 
is generated by the basis
\beq\label{eq:rep-flag-2}
\{w_1,\dots, w_{r-1},cw_r+w_{r+1},w_{r},w_{r+2},\dots, w_{r^A}\}.
\eeq

Finally, let $i=r$, $\iota=0$, $s_r=-1$ (type C). Then the superflag $\mathscr F_{\bs y^{[r]},\tlbms^{[r]}}$ 
is generated by the basis
\beq\label{eq:rep-flag-3}
\{w_1,\dots, w_{r-1},2c^2w_r+2cw_{r+1}+w_{r+2},2cw_{r}+w_{r+1},w_{r},\dots, w_{r^A}\}.
\eeq

Second, a fermionic reproduction procedure.  Assume that we are not in type D with $\kappa(\tlbms)=-1$ and $i=r-1$. 
Let $i<r=m+n$ and $s_i=-s_{i+1}$. Then the superflag $\mathscr F_{\bs y^{[i]},\tlbms^{[i]}}$ is generated by the basis
\beq\label{eq:rep-flag-4}
\{w_1,\dots, w_{i-1},w_{i+1},w_{i},w_{i+2},\dots,w_{r^A-i-1}, w_{r^A-i+1},w_{r^A-i},w_{r^A-i+2},\dots, w_{r^A}\}.
\eeq
In the case, $i=r-1$, $\iota=0$,  $s_r=1$ (type D), $\kappa(\tlbms)=-1$, the basis above generates the superflag $\mathscr F_{\bs y^{[r]},\tlbms^{[r]}}$.

\medskip

Third, the fake reproduction procedure. We have $\iota=0$, $s_r=1$. 
Then the superflag $\mathscr F_{\bs y^{[f]},\tlbms^{[f]}}$ is generated by the basis
\beq\label{eq:rep-flag-5}
\{w_1,\dots, w_{r-1},w_{r+2},w_{r+1},w_{r},w_{r+3},\dots, w_{r^A}\}.
\eeq

\subsection{The main theorem}
Now we are ready to describe our main result.

Let $\g=\osp_{2m+\iota|2n}$. Let $\bla$ be an $A$-typical sequence of dominant integral $\osp_{2m+\iota|2n}$-weights, $\bm z$ a sequence of pairwise distinct complex numbers. Let $\tlbms_0$ be an extended parity sequence in $\wt S_{m|n}^\iota$. Let $C^{\bm s_0}$ be the Cartan matrix of $\g$ associated to the parity sequence $\bm s_0$. Let $\bm T^{\tlbms_0}$ be the sequence of rational functions associated with $\bla$, $\bm z$, and $\tlbms_0$, see \eqref{eq:T-polynomials}.

Let $\bs y_0$ be a sequence of polynomials which represents a solution of the Bethe ansatz equation \eqref{eq:BAE-poly-T-nos} associated to $(C^{\bm s_0},\bs T^{\tlbms_0})$. Let $\mathscr P=\mathscr P_{\bm y_0,\tlbms_0}$ be the population of solutions of the Bethe ansatz equation originated from $(\bs y_0,\tlbms_0)$

Let $\mc R=\mathscr R^{2m+\iota|2n}_{\bm y_0,\tlbms_0}$ denote the symmetric $(2m|2n+1-\iota)$-rational pseudo-differential operator associated to $\mathscr P$. Let $W=V\oplus U$ be the superkernel of $\mc R$ and write $\mc R=\mc D_{\bar 0}\mc D_{\bar 1}^{-1}$, where $V=\ker \mc D_{\bar 0}$ and $U=\ker\mc D_{\bar 1}$.

\begin{thm}\label{main thm} Assume $\bs y_0$ is superfertile. Then $W$ is a vector superspace of rational functions. 
We have natural bijections between the population $\mathscr P$, the variety of isotropic superflags $\mathscr F^\perp(W)$, and the set $\mc F^\perp(\mc R)$ of symmetric complete factorizations of the rational pseudo-differential operator $\mc R$. 
\end{thm}
\begin{proof}
By Proposition \ref{prop:bij-f-iso}, we have a bijection between the variety of isotropic superflags $\mathscr F^\perp(W)$ and the set $\mc F^\perp(\mc R)$ of symmetric complete factorizations of $\mc R$.

By \eqref{eq:B-oper} and \eqref{eq:D-oper}, we have a map from the population $\mathscr P$ to the set of symmetric complete factorizations $\mc F^\perp(\mc R)$. In the case $\iota=1$, the map is clearly injective. 

We now show the map is injective for $\iota=0$ as well.
Combining Proposition \ref{prop:bij-f-iso} and \eqref{eq:D-oper}, we find the isotropic superflag $\mathscr F_0$ corresponding to $(\bs y_0,\tlbms_0)$.  Then we have the decomposition \eqref{eq:pm-decom} and $\mathscr F_0\in\mathscr F_+^\perp(W)$. Note that the fermionic and bosonic reproduction procedures preserve the decomposition, see \eqref{eq:rep-flag-1}, \eqref{eq:rep-flag-3}, \eqref{eq:rep-flag-4}, while the fake reproduction maps  superflags in $\mathscr F_\pm^\perp(W)$ to superflags in $\mathscr F_\mp^\perp(W)$, see \eqref{eq:rep-flag-5}. 
After restriction to the points of the population $\mathscr P$ with $\kappa=1$, the map is clearly injective. It follows that the map is injective.

The map from the population $\mathscr P$ to the set of symmetric complete factorizations $\mc F^\perp(\mc R)$ is surjective, since all isotropic superflags are obtained from any single one by changes \eqref{eq:rep-flag-1}-\eqref{eq:rep-flag-5} corresponding to the reproduction procedure.

Finally, we show that $W$ consists of rational functions. We consider a symmetric complete factorization of $\mc R$ corresponding to the parity  $(1,\dots,1,-1,\dots,-1,1,\dots,1)$. Then $\mc R=\bar{\mc D}_1^* \bar{\mc D}_2^{-1} \bar {\mc D}_1$. Note that $\bar {\mc D}_1$ and $\bar {\mc D}_2$ are not unique. Clearly, kernels of $\bar{\mc D}_1$ and $\bar {\mc D}_2$ consist of rational functions. Indeed, for example, for $\bar {\mc D}_2$ the kernel of last linear factor is explicitly a rational function and due to superfertility, it stays a rational function under reproduction procedure. The choice of the last linear factor corresponds to a choice of isotropic vector in $V$ and the isotropic vectors are dense in $V$. Then the kernel $U$ of $\mc D_{\bar 1}$ consists of rational functions by Lemma \ref{lem:exchange-pm}. 

Similarly, starting from a parity $(-1,\dots,-1,1,\dots,1,-1,\dots,-1)$, we show that 
the kernel $V$ of $\mc D_{\bar 0}$ also consists of rational functions.
\end{proof}

\begin{rem}\label{rem after thm}
If the assumption of superfertility is dropped, the same argument as in the proof of Theorem \ref{main thm} shows that the population $\mathscr P$ can be embedded into the set of isotropic superflags $\mathscr F^\perp(W)$.\qed
\end{rem}

The population $\mathscr P$, the set of symmetric complete factorizations $\mc F^\perp(\mc R)$, and the variety of isotropic superflags $\mathscr F^\perp(W)$ are naturally direct unions:
\[
\mathscr P=\bigsqcup_{\bs s}\mathscr P^{\bs s},\qquad \mc F^\perp(\mc R)=\bigsqcup \mc F_{\bs s}^\perp(\mc R),\qquad 
\mathscr F^\perp(W)=\bigsqcup\mathscr F^\perp_{\bs s}(W).
\]
Here $\mathscr P^{\bs s}$ is the set of all pairs $(\bs y,\tlbms)\in\mathscr P$ with parity $\bs s$,   $\mc F_{\bs s}^\perp(\mc R)$ is the set of all symmetric complete $\bm s^A$-factorizations of $\mc R$, and $\mathscr F^\perp_{\bs s}(W)$ is the set of all isotropic superflags in $\mathscr F_{\bs s^A}(W)$. The parity for a superflag is given by the parity of the homogeneous basis generating the superflag.

We note that the bijections in Theorem \ref{main thm} are clearly compatible with the parity decompositions.

\medskip

Recall that in type D, each solution of the Bethe ansatz equation corresponds to two points in the population $\mathscr P$, one with $\kappa=1$ and one with $\kappa=-1$. These two points are connected via the fake reproduction procedure. 

\subsection{The even cases}\label{sec:even}
Theorem \ref{main thm} in the special cases   $\osp_{2m+1|0}=\mathfrak{so}_{2m+1}$, $\osp_{0|2n}=\mathfrak{sp}_{2n}$  recovers the known theorems for even cases of types B$_m$ and C$_n$ of \cite{MV04}. 

The Lie algebra $\g=\osp_{2m|0}$ is even and is identified with the simple Lie algebra of type D$_m$. In this case, the rational pseudo-differential operator has the form $\mc R=\mc D\partial^{-1}\mc D^*$, where $\mc D$ is a differential operator of order $m$. Such an operator is natural to expect, see \cite[Proposition 8.5]{DS}, \cite[Corollary 3.8]{MM}.

Note that in this case we have only bosonic and fake reproduction procedures.

We also note that in this case we have $W=V\oplus U$ with $\dim U=1$. Moreover, due to Lemma \ref{lem:super-self-dual}, the space $U$ is spanned by $\Wr(V)$, and, in particular, it is completely determined by $V$. Recall that $V$ is a space of dimension $2m$ with a natural symmetric bilinear form. Then isotropic superflags in $W$ are in bijection with isotropic flags in $V$.

In particular, each solution of the Bethe ansatz equation corresponds to two superflags. This is in accordance with the known general fact that a population for the even case is always isomorphic to the flag variety $G^\vee/B^\vee$ of the Langlands dual Lie group $G^\vee$, see \cite{F,MV05}.

\subsection{Proof of Proposition \ref{conj true}}\label{sec proof of prop}
We prepare a lemma saying that under some assumptions the fermionic reproduction procedure in the $i$-th direction can be done in such a way that the immediate descendant in the $i$-th direction is generic, provided that at least one of the simple roots $\alpha_{i\pm 1}^{\bs s}$ is even. 

Let $\bm s\in S_{m|n}$ be such that $s_{i-1}=s_{i}\neq s_{i+1}$.

Let $\bla=(\la_1,\dots,\la_p)$ be a sequence of dominant integral typical $\glMN$-weights. Let $\bm T^{\bm s}$ be the sequence of rational functions associated to $\bla$, $\bm z$, and $\bm s$. Since $\la_j$ are typical, we always have $\pi_j^{\bs s}=\pi(T_{j}^{\bm s})=\prod_{k=1}^p(x-z_k)$, whenever the simple root $\alpha_j^{\bm s}$ is odd. Denote $\prod_{k=1}^p(x-z_k)$ by $\xi(x)$.

Let $\bm y=(y_1,\dots,y_{r-1})$ be a sequence of polynomials representing a solution of the $\glMN$ Bethe ansatz equation associated to $(C^{\bm s},\bm T^{\bm s})$, where $C^{\bm s}$ is the Cartan matrix of $\glMN$ associated to $\bm s$. 

By Theorem \ref{thm:A-rep} there exists a polynomial $\tl y_i$ such that $\Wr(y_{i-1},\tl y_{i-1})=T_{i-1}^{\bm s}y_{i-2}y_{i}$. Since $\bm y$ is generic with respect to $\bla$, $\bm s$, and $\bm z$, we have  $\tl y_{i-1}$ is relatively prime to $y_{i-1}(x)$. In particular, we can choose $\tl y_{i-1}$ such that it has no multiple roots, no common roots with $y_{i+1}$, and such that $\tl y_{i-1}(z_j)\neq 0$ for all $j=1,\dots,p$.

For $c\in\C^\times$, denote $$\bm y_c=(y_1,\dots,y_{i-1}+c\tl y_{i-1},\dots, y_{r-1}),$$
the immediate descendants of $\bm y$ in the $(i-1)$-st direction depending on a parameter $c$.

\begin{lem}\label{lem: generic fermionic}
Assume $y_{i-2},y_{i+1}$ in $\bm y$ have no multiple roots. Assume $y_{i-1},y_{i+1}$ in $\bm y$  do not vanish at $z_j$, for each $j=1,\dots,p$. Then for generic $c\in \C$, the result of fermionic reproduction procedure $(\bm y_c)^{[i]}$ is generic with respect to $\bla$, $\bm s^{[i]}$, and $\bm z$, and all polynomials in $(\bm y_c)^{[i]}$ have no multiple roots. 
\end{lem}
\begin{proof}
It follows from Theorem \ref{thm:A-rep} there exists a polynomial $\tl y_{i}$ such that
\beq\label{eq:last 1}
y_{i}\tl y_{i}=\ln'\Big(\frac{T_{i}^{\bm s}y_{i-1}}{y_{i+1}}\Big)\xi(x)y_{i-1}y_{i+1}.
\eeq
In particular, we conclude $\tilde y_i(z_j)\neq 0$ since the right hand side does not vanish at $z_j$. 

Similarly, there exists a polynomial $\bar y_{i}$ such that
\beq\label{eq:last 2}
y_{i}\bar y_{i}=\ln'\Big(\frac{T_{i}^{\bm s}\tl y_{i-1}}{y_{i+1}}\Big)\xi(x)\tl y_{i-1}y_{i+1}.
\eeq
We observe that $\bar y_i(z_j)\neq 0$ and $\bar y_i$ has no common roots with $y_{i+1}$.

Multiplying \eqref{eq:last 1} by $\tl y_{i-1}$,   \eqref{eq:last 2} by $y_{i-1}$, subtracting the results and dividing by $y_i$, we obtain
\beq\label{eq:last 3}
\tl y_{i}\tl y_{i-1}-\bar y_{i}y_{i-1}=\xi(x)T_{i-1}^{\bm s}y_{i-2}y_{i+1}.
\eeq
Then $\tl y_{i}$ and $\bar y_{i}$ have no common multiple roots, since all common roots are also roots of $y_{i-2}$ which has no multiple roots.
Therefore, for generic $c$, $\tl y_{i}+c\bar y_{i}$ has no multiple roots and is relatively prime to $y_{i+1}$. 

Clearly, we have 
\[
(\bm y_c)^{[i]}=(y_1,\dots,y_{i-2},y_{i-1}+c\tl y_{i-1},\tl y_{i}+c\bar y_{i},y_{i+1},\dots,y_{r-1}).
\]
Since $y_{i-1}$ and $\tl y_{i-1}$ are relatively prime, for generic $c$,
$y_{i-1}+c\tl y_{i-1}$ is relatively prime to $y_{i+1}$ and $y_{i-2}$ and has no multiple roots.

Finally, we see that for generic $c$, $y_{i-1}+c\tl y_{i-1}$ is also relatively prime to $\tl y_{i}+c\bar y_{i}$ using the equation
\[
y_{i}(\tl y_{i}+c\bar y_{i})=\ln'\Big(\frac{T_{i}^{\bm s}(y_{i-1}+c\tl y_{i-1})}{y_{i+1}}\Big)\xi(x)(y_{i-1}+c\tl y_{i-1})y_{i+1},
\]
which is the sum of \eqref{eq:last 1}, \eqref{eq:last 2}.
\end{proof}
In Lemma \ref{lem: generic fermionic} we made assumptions on the roots of $y_{i-2},y_{i-1},y_{i+1}$. The assumption on $y_{i-1}$ in fact can be dropped by choosing a different member of the $\bs y_c$ family. The same is true for $y_{i-2}$ and $y_{i+1}$ if the corresponding simple roots are even. If one of the simple roots is odd then the corresponding polynomial does not vanish at $z_j$ automatically since all weights are typical. Thus the only essential assumption is that polynomials $y_{i-2},y_{i+1}$ corresponding to odd simple roots have no multiple roots.

\begin{proof}[Proof of Proposition \ref{conj true}] We consider the case $\osp_{2m+1|2}$. 
If $\bs y$ is superfertile, then the population has $m+1$ components corresponding to parity sequences, each of dimension $m^2+1$. Since, fertility is a closed condition, see \cite[Lemma 3.6]{MV04}, it is enough to find a family of fertile tuples in the population originated at $\bs y$ of dimension $m^2+1$ in each component. Indeed, the population is embedded into the variety of isotropic superflags $\mathscr F^\perp(W)$, see Remark \ref{rem after thm} and in this case, it has full dimension in each irreducible component of $\mathscr F^\perp(W)$. Therefore it coincides with $\mathscr F^\perp(W)$. It follows that the population is closed under all reproduction  procedures, and each point is fertile.

We start with the parity sequence $\bm s_+$. Then $\bs y^A=(y_1,\dots,y_{m},y_{m+1},y_{m},\dots,y_{1})$ is a solution of the $\gl_{2m|2}$ Bethe ansatz equations of parity
$(1,\dotsm,1,-1,-1,1,\dots,1)$. First, we do the bosonic reproduction procedure in the direction $m+1$ which gives us a one-parameter family of generic sequences $\bs y^{[m+1]}$ which are fertile by Theorem \ref{thm:A-rep}. Then we use Lemma \ref{lem: generic fermionic} to perform a sequence of pairs of fermionic reproduction procedures to $(\bs y^{[m+1]})^A$ in the directions $(m,m+2), (m-1,m+3),\dots,(2,2m)$ keeping the sequence generic and symmetric. Then we do reproduction procedures in directions $1$ and $2m+1$. Let us call the result $\tilde{\bs y}^A$. The sequence $\tilde {\bs y}^A=(\tl y_i^A)_{1\lle i\lle 2m+1}$ depends on one parameter and corresponds to parity $(-1,1,\dots,1,-1)$. Then  we have $\tilde y_{i}^A=\tilde y_{2m+2-i}^A$ and
$(\tilde y_2^A,\dots, \tilde y_{m+1}^A)$ is a solution of the Bethe ansatz equation of type B$_m$, cf. Lemma \ref{lem:AB}. Performing the bosonic reproduction procedures we obtain a family of generic sequences of dimension $m^2$, see \cite{MV04}.

Finally, we use Lemma \ref{lem: generic fermionic}, to do fermionic reproduction procedures in pairs of directions $(1,2m+1)$, $(2,2m)$, $\dots$, $(m,m+2)$ to obtain a family of generic sequences in all parities.

The case of $\osp_{3|2n}$ is similar.
\end{proof}

\end{document}